\documentclass{llncs} %{article}
\usepackage{enumitem,amsmath,amssymb}

\usepackage{tikz}
\usetikzlibrary{calc,positioning,arrows}
\def\centerarc[#1](#2)(#3:#4:#5)% Syntax: [draw options] (center) (initial angle:final angle:radius)
    { \draw[#1] ($(#2)+({#5*cos(#3)},{#5*sin(#3)})$) arc (#3:#4:#5); }

\usepackage{graphicx,xcolor}
\usepackage{float}
\usepackage{hyperref}
\hypersetup{
  colorlinks   = true, %Colours links instead of ugly boxes
  urlcolor     = blue, %Colour for external hyperlinks
  linkcolor    = blue, %Colour of internal links
  citecolor   = red %Colour of citations
}

\usepackage[linesnumbered,boxed,ruled]
       {algorithm2e} %[linenumbered]
\SetAlgoInsideSkip{medskip}

\usepackage{todonotes}
\usepackage[skip=2.2pt,
           ]{parskip}

\newtheorem{observation}{Observation}
\newtheorem{fact}{Fact}

{\bf}{\it}
{\bf}{\it}

\newcommand{\NP}{\mathsf{NP}}

\def\NAE3SAT{{\sc nae 3-sat}}
\def\PNAE3SAT{{\sc positive nae 3-sat}}
\def\1IN3SAT{{\sc positive 1-in-3 3-sat}}
\def\P1IN3SAT{{\sc positive planar 1-in-3 3-sat}}

\newcommand{\VC}{\textsc{vertex cover}}
\newcommand{\CVD}{\textsc{cluster-vd}}
\newcommand{\clawVD}{\textsc{claw-vd}}
\newcommand{\dclawVD}[1]{\textsc{{$#1$}-claw-vd}}
\newcommand{\dHVC}[1]{\textsc{{$#1$}-hvc}}
\newcommand{\minVC}{\textsc{min vertex cover}}
\newcommand{\minCVD}{\textsc{min cluster-vd}}

\newcommand{\mindclawVD}[1]{\textsc{min {$#1$}-claw-vd}}
\newcommand{\mindHVC}[1]{\textsc{min {$#1$}-hvc}}

\begin{document}

\title{
On the $d$-Claw Vertex Deletion Problem\thanks{A preliminary version~\cite{HsiehLP21} of this article has appeared in Proceedings of the 27th International Computing and Combinatorics Conference (COCOON 2021); 
%the second author discovered the erroneous proof of \cite[Theorem 3]{HsiehLP21} and corrected it. 
the second author discovered and corrected the incorrect proof of Theorem~3 in~\cite{HsiehLP21}.} 
}

\author{Sun-Yuan Hsieh\inst{1} \and Hoang-Oanh Le\inst{2} \and Van Bang Le\inst{3} \and Sheng-Lung Peng\inst{4}\thanks{Corresponding author.}
}
\institute{
Department of Computer Science and Information Engineering, National Cheng Kung University, Tainan 70101, Taiwan\\
\email{hsiehsy@mail.ncku.edu.tw}
\and
Berlin, Germany\\
\email{LeHoangOanh@web.de}
\and
Institut f\"{u}r Informatik, Universit\"{a}t Rostock, Rostock, Germany\\
\email{van-bang.le@uni-rostock.de}
\and
Department of Creative Technologies and Product Design,
National Taipei University of Business, Taoyuan 32462, Taiwan\\
\email{slpeng@ntub.edu.tw}
}

\maketitle

\thispagestyle{plain} %page numer with llncs

\setcounter{footnote}{0} %footnote zur�cksetzen wegen footnote bei den Autoren

%%%%%%%%%%%%%%%%%%%%%%%%%%

\begin{abstract}
Let \emph{$d$-claw} (or \emph{$d$-star}) stand for $K_{1,d}$, the complete bipartite graph with~$1$ and $d\ge 1$~vertices on each part. The $d$-claw vertex deletion problem, \dclawVD{d}, asks for a given graph $G$ and an integer $k$ if one can delete at most~$k$ vertices from $G$ such that the resulting graph has no $d$-claw as an induced subgraph. Thus, \dclawVD{1} and \dclawVD{2} are just the famous \VC\ problem and the \textsc{cluster vertex deletion} problem, respectively.

In this paper, we strengthen a hardness result in [M.\!~Yannakakis, Node-Deletion Problems on Bipartite Graphs, SIAM J. Comput. (1981)], by showing that \textsc{cluster vertex deletion} remains $\NP$-complete when restricted to bipartite graphs of maximum degree~3. %of diameter~3. 
Moreover, for every $d\ge 3$, we show that \dclawVD{d} is $\NP$-complete even when restricted to bipartite graphs of maximum degree~$d$. These hardness results are optimal with respect to degree constraint.
By extending the hardness result in [F.\!~Bonomo-Braberman \emph{et al.}, Linear-Time Algorithms for Eliminating Claws in Graphs, COCOON 2020], we show that, for every $d\ge 3$, \dclawVD{d} is $\NP$-complete even when restricted to split graphs without $(d+1)$-claws, and split graphs of diameter~2.
On the positive side, we prove that \dclawVD{d} is polynomially solvable on what we call $d$-block graphs, a class properly contains all block graphs. This result extends the polynomial-time algorithm in [Y.\!~Cao \emph{et al.}, Vertex deletion problems on chordal graphs, Theor. Comput. Sci. (2018)] for \dclawVD{2} on block graphs to \dclawVD{d} for all $d\ge 2$ and improves the polynomial-time algorithm proposed by F.~Bonomo-Brabeman \emph{et al.} for (unweighted) \dclawVD{3} on block graphs to $3$-block graphs.
\end{abstract}

 \textbf{keyword:}\\
Vertex Cover, Cluster Vertex Deletion, Claw Vertex Deletion, Graph Algorithm, NP-complete Problem
 %\end{keyword}

\section{Introduction}
%---------------------
Graph modification problems are a very extensively studied topic in graph algorithm. One important class of graph modification problems is as follows.
Let $H$ be a fixed graph.
The \textsc{$H$ vertex deletion} (\textsc{$H$-vd} for short) problem takes as input a graph $G$ and an integer~$k$. The question is whether it is possible to delete a vertex set $S$ of at most $k$ vertices from $G$ such that the resulting graph is \emph{$H$-free}, \emph{i.e.}, $G-S$ contains no induced subgraphs isomorphic to $H$. The optimization version asks for such a vertex set~$S$ of minimum size, and is denoted by \textsc{min $H$ vertex deletion} (\textsc{min $H$-vd} for short).

The case that $H$ is a $2$-vertex path, \emph{i.e.}, an edge, is the famous \VC\ problem, one of the basic $\NP$-complete problems.
The case that $H$ is a $3$-vertex path is well known under the name \textsc{cluster vertex deletion} (\CVD\ for short).
Very recently, the COCOON 2020 paper~\cite{Bonomo-Braberman20} addresses the case that $H$ is the \emph{claw} $K_{1,3}$, the complete bipartite graph with~$1$ and~$3$ vertices in each part, thus the \clawVD\ problem.

For any integer $d>0$, let \emph{$d$-claw} (or \emph{$d$-star}) stand for $K_{1,d}$, the complete bipartite graph with~$1$ and~$d$ vertices on each part. In this paper, we go on with the \clawVD\ problem by considering the \dclawVD{d} problem for any given integer~$d>0$:

\medskip\noindent
\fbox{
\begin{minipage}{.96\textwidth}
\dclawVD{d}\\[.7ex]
\begin{tabular}{l l}
{\em Instance:\/}& A graph $G=(V,E)$ and an integer $k<|V|$.\\
{\em Question:\/}& Is there a subset $S\subset V$ of size at most $k$ such that $G-S$ is $d$-claw free\,?\\%[.7ex]
\end{tabular}
\end{minipage}
}

\medskip\noindent
Thus, \dclawVD{1} and \dclawVD{2} are just the well-known $\NP$-complete problems \VC\ and \CVD, respectively, and
\dclawVD{3} is the \clawVD\ problem addressed in the recent paper~\cite{Bonomo-Braberman20} mentioned above.

While \dclawVD{1} is polynomially solvable when restricted to perfect graphs (including chordal and bipartite graphs)~\cite{GLS1988},
\dclawVD{d} is $\NP$-complete for any $d\ge 2$ even when restricted to bipartite graphs~\cite{Yannakakis81a}. %\cite{LewisY80,Yannakakis81a}.
When restricted to chordal graphs, it is shown in~\cite{Bonomo-Braberman20} that \dclawVD{3} remains $\NP$-complete even on split graphs. The computational complexity of \dclawVD{2} on chordal graphs is still unknown~\cite{CaoKOY17,0001KOY18}.
Both \dclawVD{2} and \dclawVD{3} can be solved in polynomial time on block graphs~\cite{Bonomo-Braberman20,0001KOY18}, a proper subclass of chordal graphs containing all trees.

%\paragraph{Known results and related work.}
It is well known that the classical $\NP$-complete problem \VC\ remains hard when restricted to planar graphs of maximum degree~$3$ and arbitrary large girth. %~cite{}.
It is also known that, assuming ETH (Exponential Time Hypothesis), \VC\ admits no subexponential-time algorithm in the vertex number~\cite{LokshtanovMS11}
and, while \minVC\ can be approximated within factor~$2$ by a simple \lq textbook\rq\ greedy algorithm, no polynomial-time approximation with a factor better than~2 exists assuming UGC (Unique Games Conjecture)~\cite{KhotR08}. 
%, \minVC\ cannot be approximated by a factor better than~$2$~\cite{KhotR08} while the simple greedy algorithm guarantees an approximation ratio of~$2$. 
Very recently, it is shown in~\cite{AprileDFH22} that \minCVD\ can be approximated within factor~$2$, and this is optimal assuming UGC.

As for \minVC, \mindclawVD{d} can be approximated within a factor $d+1$ but there is no polynomial-time approximation scheme~\cite{LundY93}.
From the results in~\cite{KumarMDS14} it is known that, for any $d\ge 2$, \mindclawVD{d} admits a $d$-approximation algorithm on bipartite graphs.
This result was improved later by a result in~\cite{GuruswamiL17}, where the related problem \textsc{$d$-claw-transversal} was considered. Given a graph $G$, this problem asks to find a smallest vertex set $S\subseteq V(G)$ such that $G-S$ does not contain a $d$-claw as a (not necessarily induced) subgraph. In~\cite{GuruswamiL17}, it was shown that, in contrast to our \mindclawVD{d} problem, \textsc{$d$-claw-transversal} can be approximated within a factor of $O(\log (d+1))$. Since \dclawVD{d} and \textsc{$d$-claw-transversal} coincide when restricted to bipartite graphs, \dclawVD{d} admits an $O(\log(d+1))$-approximation on bipartite graphs.
The \textsc{2-claw-transversal} is also known as \textsc{$P_3$ vertex cover} (see, \emph{e.g.},~\cite{ChangCHRS16,Tsur19a}).

By a standard bounded search tree technique, \dclawVD{d} admits a parameterized algorithm running in $O^*((d+1)^k)$ time\footnote{The $O^*$ notation hides polynomial factors.}.  The current fastest parameterized algorithm for \VC\ and \CVD\ has runtime $O^*(1.2738^k)$~\cite{ChenKX10} and $O^*(1.811^k)$~\cite{Tsur21}, respectively.
% It is known that \VC\ (respectively, \CVD) can be approximated within a factor 2 (respectively, $\frac{9}{4}$~\cite{FioriniJS20}) but there is no polynomial-time approximation scheme~\cite{LundY93}. 

For the edge modification versions, there is a comprehensive survey~\cite{abs-2001-06867}.

%\paragraph{Our contributions.} 
In this paper, we first derive some hardness results by a simple reduction from \VC\ to \dclawVD{d}, stating that \dclawVD{d} does not admit a subexponential-time algorithm in the vertex number unless the ETH fails, and that \dclawVD{d} remains $\NP$-complete when restricted to planar graphs of maximum degree~$d+1$ and arbitrary large girth.

We then revisit the case of bipartite input graphs by showing that \CVD\ remains $\NP$-complete on bipartite graphs of maximum degree~$3$, and for $d\ge 3$, \dclawVD{d} remains $\NP$-complete on bipartite graphs of maximum degree~$d$ and on bipartite graphs of diameter~$3$.
These hardness results for \dclawVD{d} are optimal with respect to degree and diameter constraints, and improve the corresponding hardness results for \dclawVD{d}, $d\ge 2$, on bipartite graphs in~\cite{Yannakakis81a}.

Further, we extend the hardness results in~\cite{Bonomo-Braberman20} for \clawVD\ to \dclawVD{d} for every $d\ge 3$.
We show that \dclawVD{d} is $\NP$-complete even when restricted to split graphs without $(d+1)$-claws and, assuming the UGC, it is hard to approximate \mindclawVD{d} to a factor better than $d-1$.
We obtain these hardness results by modifying the reduction from \VC\ to \clawVD\ given in~\cite{Bonomo-Braberman20} to a reduction from \textsc{hypergraph vertex cover} on to \dclawVD{d}.

We complement the negative results by showing that \dclawVD{d} is polynomial-time solvable on what we call $d$-block graphs, a class that contains all block graphs. As block graphs are $2$-block graphs, and $d$-block graphs are ($d+1$)-block graphs but not the converse, our positive result extends the polynomial-time algorithm for \dclawVD{2} on block graphs in~\cite{0001KOY18} to \dclawVD{d} for all $d\ge 2$, and for \dclawVD{3} on block graphs in~\cite{Bonomo-Braberman20} to $3$-block graphs.

The paper is organized as follows. In Section~\ref{sec:pre}, we give some notation and terminologies used in this paper. Section~\ref{sec:hard} presents hardness results mentioned above. A polynomial result on $d$-block graphs is shown in Section~\ref{sec:d-block}. 
Section~\ref{sec:conclusion} concludes the paper with some remarks. 

\section{Preliminaries}\label{sec:pre}
%=====================================
%\paragraph{Notation and Terminology.}
Let $G=(V,E)$ be a graph with vertex set $V$ and edge set $E$.
The neighborhood of a vertex $v$ in $G$, denoted by $N_G(v)$, is the set of all vertices in $G$ adjacent to $v$; if the context is clear, we simply write $N(v)$.
Let $\deg(v) = |N(v)|$ to denote the degree of the vertex $v$. The close neighborhood of $v$ is denoted by $N[v]$ that is $N(v)\cup \{v\}$. 

We call a vertex \emph{universal} if it is adjacent to all other vertices. Vertices of degree~$1$ are called \emph{leaves}.
The \emph{distance} between two vertices in $G$ is the length of a shortest path connecting the two vertices, the \emph{diameter} is the maximal distance between any two vertices, and the \emph{girth} is the length of a shortest cycle in $G$ (if exists).

A \emph{center vertex} of a $d$-claw $H$ is a universal vertex of $H$; if $d\ge 2$, the center of a $d$-claw is unique.
We say that a $d$-claw is \emph{centered at vertex $v$} if $v$ is the center vertex of that $d$-claw. 
Note that the $2$-claw $K_{1,2}$ and the $3$-vertex path $P_3$ coincide, and that the $3$-claw $K_{1,3}$ is also called \emph{claw}. 

An \emph{independent set} (respectively, a \emph{clique}) in a graph $G=(V,E)$ is a vertex subset of pairwise non-adjacent (respectively, adjacent) vertices in $G$. Graph $G$ is a \emph{split graph} if its vertex set $V$ can be partitioned into an independent set and a clique.
A~graph is a \emph{cluster graph} if it is a vertex disjoint union of cliques. Equivalently, cluster graphs are exactly those without induced $3$-vertex path $P_3$. 

For a subset~$S\subseteq V$, $G[S]$ is the subgraph of $G$ induced by $S$, and $G-S$ stands for $G[V\setminus S]$. For simplicity, for a set $S$ and an element $v$, we use $S+v$ (respectively, $S-v$) to denote $S\cup \{v\}$ (respectively, $S\setminus \{v\}$). 

Let $H$ be a fixed graph. An \emph{$H$-deletion set} is a vertex set $S\subseteq V(G)$ such that $G-S$ is $H$-free.  A $K_{1,1}$-deletion set and a $K_{1,2}$-deletion set are known as \emph{vertex cover} and \emph{cluster deletion set}, respectively. In other words, the cluster vertex deletion problem is the problem of finding a minimum cluster deletion set $S$ on $G$ such that the resulting graph $G-S$ is a cluster graph.

A \emph{hypergraph} $G=(V,E)$ consists of a vertex set $V$ and an edge set $E$ where each edge $e\in E$ is a subset of $V$. Let $r\ge 2$ be an integer. A hypergraph is \emph{$r$-uniform} if each of its edges is an $r$-element set. Thus, a $2$-uniform hypergraph is a graph in usual sense.
A \emph{vertex cover} in a hypergraph $G=(V,E)$ is a vertex set $S\subseteq V$ such that $S\cap e\not=\emptyset$ for any edge $e\in E$.
The \textsc{$r$-hypergraph vertex cover} (\dHVC{r} for short) problem asks, for a given $r$-uniform hypergraph $G=(V,E)$ and an integer $k<|V|$, whether $G$ has a vertex cover $S$ of size at most $k$.
The optimization version asks for such a vertex set~$S$ of minimum size and is denoted by \mindHVC{r}. Note that \dHVC{2} and \mindHVC{2} are the famous \VC\ problem and \minVC\ problem, respectively. It is known that \dHVC{r} is $\NP$-complete and \mindHVC{r} is UGC-hard to approximate within a factor better than $r$~\cite{KhotR08}.

\section{Hardness results}\label{sec:hard}
%=========================================
Recall that \dclawVD{d} is $\NP$-complete even on bipartite graphs~\cite{Yannakakis81a}. 
We begin with two simple observations which lead to further hardness results on other restricted graph classes.

\begin{observation}\label{obs:diam2}
\dclawVD{d} remains $\NP$-complete on graphs of diameter~$2$.
\end{observation}
\begin{proof}
Given an instance $(G,k)$ for \dclawVD{d}, let $G'$ be obtained from $G$ by adding a $d$-claw with center vertex~$v$ and joining $v$ to all vertices in $G$. Then $v$ is a universal vertex in $G'$ and hence $G'$ has diameter~$2$. Moreover, $(G,k)\in\text{\dclawVD{d}}$ if and only if $(G',k+1)\in\text{\dclawVD{d}}$.\qed
\end{proof}

We remark that the graph $G'$ in the proof above is a split graph whenever $G$ is a split graph, and $G'$ has only one vertex of unbounded degree whenever $G$ has bounded maximum degree. The bipartite version of Observation~\ref{obs:diam2} is:
\begin{observation}\label{obs:diam3}
For any $d\ge 2$, \dclawVD{d} remains $\NP$-complete on bipartite graphs of diameter~$3$.
\end{observation}
\begin{proof}
Let $(G,k)$ be an instance for \dclawVD{d}, where $G=(X,Y,E)$ is a bipartite graph.
Let $G'$ be the bipartite graph obtained from $G$ by adding two $d$-claws with center vertices~$x$ and~$y$, respectively, and joining $x$ to all vertices in $Y\cup\{y\}$ and $y$ to all vertices in $X\cup\{x\}$. Then $G'$ has diameter~$3$.
Moreover, $(G,k)\in\text{\dclawVD{d}}$ if and only if $(G',k+2)\in\text{\dclawVD{d}}$.\qed
\end{proof}
We remark that the bipartite graph $G'$ in the proof above has only two vertices of unbounded degree whenever $G$ has bounded maximum degree.

We now describe a simple reduction from \VC\ to \dclawVD{d} and some implications for the hardness of \dclawVD{d}.
Let $d\ge 2$. Given a graph $G=(V,E)$, construct a graph $G'=(V',E')$ as follows.
\begin{itemize}
\item for each $v\in V$ let $I(v)$ be an independent set of $d-1$ new vertices;
\item $V'=V\cup\,\bigcup_{v\in V} I(v)$;
\item $E'=E\cup\,\bigcup_{v\in V} \{vx \mid x\in I(v)\}$.
\end{itemize}
Thus, $G'$ is obtained from $G$ by attaching to each vertex $v$ a set $I(v)$ of $d-1$ leaves.

\smallskip
\begin{fact}\label{Claim4}
If $S$ is a vertex cover in $G$, then $S$ is a $d$-claw deletion set in $G'$.
\end{fact}
\begin{proof} This follows immediately from the construction of $G'$. Indeed, since $G-S$ is edgeless, every connected component of $G'-S$ is a single vertex (from $I(v)$ for some $v\in S$) or a $(d-1)$-claw (induced by $v$ and $I(v)$ for some $v\not\in S)$. Thus, $S$ is a $d$-claw deletion set in~$G'$.\qed
\end{proof}

\begin{fact}\label{Claim5}
If $S'$ is a $d$-claw deletion set in $G'$, then $G$ has a vertex cover $S$ with $|S|\le |S'|$.
\end{fact}
\begin{proof} Let $S'$ be a $d$-claw deletion set in $G'$. We may assume that, for every $v\in V(G)$, $S'$ contains no vertex in $I(v)$. Otherwise, $(S'\setminus I(v))\cup\{v\}$ is also a $d$-claw deletion set in $G'$ of size at most $|S'|$. Thus $S'\subseteq V(G)$ and $S=S'$ is a vertex cover in $G$. For otherwise, if $uv$ were an edge in $G-S$ then $v$ and $\{u\}\cup I(v)$ would induce a $d$-claw in $G'-S=G'-S'$ centered at $v$.\qed
\end{proof}

We now derive other hardness results for \dclawVD{d} %and \mindclawVD{d}
from the previous reduction.
\begin{theorem}\label{thm:3}
Let $d\ge 2$ be a fixed integer. Assuming ETH, there is no $O^*(2^{o(n)})$ time algorithm for \dclawVD{d} on $n$-vertex graphs, even on graphs of diameter~$2$.
%, even when restricted to diameter-$2$ graphs.
\end{theorem}
\begin{proof} By Facts~\ref{Claim4} and~\ref{Claim5}, and the known fact that, assuming ETH, there is no $O^*(2^{o(n)})$ time algorithm for \VC\ on $n$-vertex graphs~\cite{LokshtanovMS11}. Since the graph $G'$ in the construction has $|V'|=|V|+(d-1)|V|=O(|V|)$ vertices, we obtain that there is no $O^*(2^{o(n)})$ time algorithm for \dclawVD{d}, too.
By (the proof of) Observation~\ref{obs:diam2}, the statement also holds for graphs of diameter~$2$.
\qed
\end{proof}

\begin{theorem}\label{thm:4}
Let $d\ge 2$ be a fixed integer.
\begin{itemize}
\item[\em (i)] \dclawVD{d} is $\NP$-complete, even when restricted to planar graphs of maximum degree $d+1$ and arbitrary large girth. %~$\nobreak >\nobreak g$.
\item[\em (ii)] \dclawVD{d} is $\NP$-complete, even when restricted to diameter-$2$ graphs with only one vertex of unbounded degree.
\end{itemize}
\end{theorem}
\begin{proof} It is known (and it can be derived, \emph{e.g.}, from~\cite{HortonK93,Murphy92}) that \VC\ remains $\NP$-complete on planar graphs $G$ of maximum degree~$3$ and arbitrary large girth, and in which the neighbors of any vertex of degree~3 in $G$ have degree~2.

Given such a graph~$G$, let $G'$ be obtained from $G$ by attaching, for every vertex $v$ of degree~$2$, $d-1$ leaves to~$v$. Then $G'$ is planar, has maximum degree $d+1$ and arbitrary large girth.
Moreover, similarly to Facts~\ref{Claim4} and~\ref{Claim5}, it can be seen that $G$ has a vertex cover of size at most $k$ if and only if $G'$ has a $d$-claw deletion set of size at most $k$. This proves~(i).
Part~(ii) follows from~(i) and the reduction in the proof of Observation~\ref{obs:diam2}.
\qed
\end{proof}
We remark that the hardness result stated in Theorem~\ref{thm:4}~(ii) is optimal in the sense that graphs of bounded diameter and bounded vertex degree have bounded size. Hence \dclawVD{d} is trivial when restricted to such graphs.

Note that \dclawVD{d} is trivial on graph of maximum degree less than $d$ (because such graphs contain no $d$-claws).
Moreover, \CVD\ is easily solvable on graphs of maximum degree~$2$. Thus, with Theorem~\ref{thm:4}~(i), the computational complexity of \dclawVD{d}, $d\ge 3$, on graphs of maximum degree $d$ remains to discuss.
We will show in the next subsections that the problem is still hard even on bipartite graphs of maximum degree~$d$.

\subsection{Bipartite graphs of bounded degree}\label{subsec:bip}
%===============================================================
Recall that \VC\ is polynomially solvable on bipartite graphs, hence previous results reported above cannot be stated for bipartite graphs.

In this subsection, we first give a polynomial reduction from \PNAE3SAT to \CVD\ showing that \CVD\ is $\NP$-complete even when restricted to bipartite graphs of degree~$3$. % and to bipartite graphs of diameter~$3$. 
Then, we give another polynomial reduction from \PNAE3SAT to \clawVD\ showing that \clawVD\ is $\NP$-complete even when restricted to bipartite graphs of maximum degree $3$. From this, the hardness of \dclawVD{d} in bipartite graphs of maximum degree~$d$ will be easily derived for any $d> 3$. 
Thus, we obtain an interesting dichotomy for all $d\ge 3$: \dclawVD{d} is polynomial-time solvable on graphs of maximum degree less than~$d$ and $\NP$-complete otherwise. 

%\textcolor{red}{Here:} We remark that the reduction in prooving the same results given in the conference version \cite[Theorem 3]{HsienLP21} is wrong as detected by the second author; the correct version given here is due to her idea.

Recall that an instance for \PNAE3SAT\ is a \textsc{3-cnf} formula $F=C_1\land C_2\land\cdots\land C_m$ over $n$ variables $x_1, x_2, \ldots, x_n$, in which each clause $C_j$ consists of three distinct variables. The problem asks whether there is a truth assignment of the variables such that every clause in $F$ has a true variable and a false variable. Such an assignment is called \emph{nae assignment}. 
It is well known that \PNAE3SAT\ is $\NP$-complete. 

\subsubsection{Cluster Vertex Deletion is hard in subcubic bipartite graphs}
%--------------------------------------------------------------------------

Our reduction is inspired by a reduction from \NAE3SAT to \CVD\ on bipartite graphs in~\cite{Yannakakis81a}.

Let $F=C_1\land C_2\land\cdots\land C_m$ over $n$ variables $x_1, x_2, \ldots, x_n$, in which each clause $C_j$ consists of three distinct variables. 
We will construct a subcubic bipartite graph $G$ such that $G$ has a cluster deletion set of size at most $2mn+11m$ if and only if $F$ admits a nae assignment. The graph $G$ contains a gadget $G(v_i)$ for each variable $v_i$ and a gadget $G(C_j)$ for each clause $C_j$. 

\textbf{Variable gadget.} For each variable $v_i$ we introduce $m$ pairs of \emph{variable vertices} $v_{ij}$ and~$v_{ij}'$ one pair for each clause $C_j$, $1\le j\le m$, as follows. First, take a cycle with $2m$ vertices $v_{i1}$, $v_{i1}'$, $v_{i2}$, $v_{i2}'$,  \ldots, $v_{im}$, $v_{im}'$ and edges $v_{i1}v_{i1}'$, $v_{i1}'v_{i2}$, $v_{i2}v_{i2}'$, \ldots, $v_{i(m-1)}'v_{im}$, $v_{im}v_{im}'$ and~$v_{im}'v_{i1}$.  
Then subdivide every edge $v_{ij}'v_{i(j+1)}$ with $4$ new vertices $w_{ij}$, $x_{ij}$, $y_{ij}$ and $z_{ij}$ to obtain a cycle on $6m$ vertices. 

The obtained graph is denoted by $G(v_i)$. The case $m=3$ is shown in Fig.~\ref{fig:CVDvi}. 

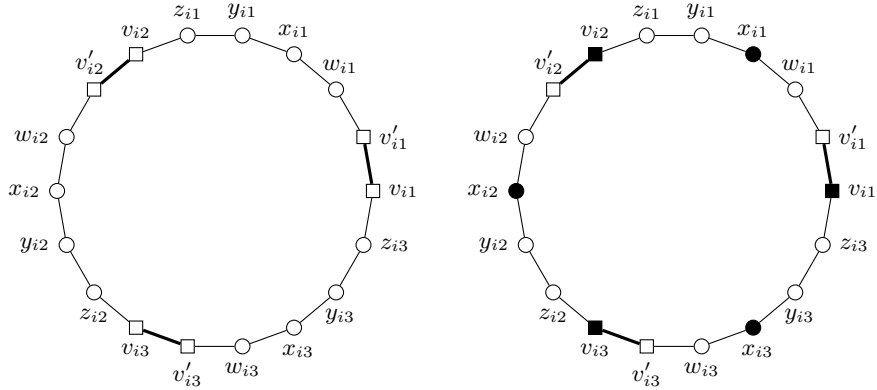
\begin{figure}[!ht]
\begin{center}
\tikzstyle{vertexS}=[draw,circle,inner sep=2pt,fill=black]
\tikzstyle{squareS}=[draw,rectangle,inner sep=2.5pt,fill=black]
\tikzstyle{subdS}=[draw,circle,inner sep=1.5pt,fill=black]
\tikzstyle{square}=[draw, rectangle,inner sep=2.5pt]
\tikzstyle{vertex}=[draw,circle,inner sep=2pt] 
\tikzstyle{subd}=[draw,circle,inner sep=1.5pt]

\begin{tikzpicture}[scale=.42] 
%\draw [color=gray!50] (0,0) grid (15,15);

\node[square] (vi1) at ({3+5*cos(0)},{4+5*sin(0)}) [label=right:\small $v_{i1}$] {};
\node[square] (vi1') at ({3+5*cos(20)},{4+5*sin(20)}) [label=right:\small $v_{i1}'$] {};
\node[square] (vi2) at ({3+5*cos(120)},{4+5*sin(120)}) [label=above:\small $v_{i2}$] {};
\node[square] (vi2') at ({3+5*cos(140)},{4+5*sin(140)}) [label=above:\small $v_{i2}'~$] {};
\node[square] (vi3) at ({3+5*cos(240)},{4+5*sin(240)}) [label=below:\small $v_{i3}$] {};
\node[square] (vi3') at ({3+5*cos(260)},{4+5*sin(260)}) [label=below:\small $v_{i3}'$] {};

\node[vertex] (wi1) at ({3+5*cos(40)},{4+5*sin(40)}) [label=above:\small $~w_{i1}$] {};
\node[vertex] (xi1) at ({3+5*cos(60)},{4+5*sin(60)})  [label=above:\small $x_{i1}$] {};
\node[vertex] (yi1) at ({3+5*cos(80)},{4+5*sin(80)}) [label=above:\small $y_{i1}$] {};
\node[vertex] (zi1) at ({3+5*cos(100)},{4+5*sin(100)}) [label=above:\small $z_{i1}$] {};
\node[vertex] (wi2) at ({3+5*cos(160)},{4+5*sin(160)})  [label=left:\small $w_{i2}$] {};
\node[vertex] (xi2) at ({3+5*cos(180)},{4+5*sin(180)}) [label=left:\small $x_{i2}$] {};
\node[vertex] (yi2) at ({3+5*cos(200)},{4+5*sin(200)}) [label=left:\small $y_{i2}$] {};
\node[vertex] (zi2) at ({3+5*cos(220)},{4+5*sin(220)}) [label=below:\small $z_{i2}$] {};
\node[vertex] (wi3) at ({3+5*cos(280)},{4+5*sin(280)}) [label=below:\small $w_{i3}$] {};
\node[vertex] (xi3) at ({3+5*cos(300)},{4+5*sin(300)}) [label=below:\small $~x_{i3}$] {};
\node[vertex] (yi3) at ({3+5*cos(320)},{4+5*sin(320)}) [label=below:\small $~y_{i3}$] {};
\node[vertex] (zi3) at ({3+5*cos(340)},{4+5*sin(340)}) [label=right:\small $z_{i3}$] {};

\draw [very thick] (vi1)--(vi1');
\draw (vi1')--(wi1)--(xi1)--(yi1)--(zi1)--(vi2);
\draw[very thick] (vi2)--(vi2');
\draw (vi2')--(wi2)--(xi2)--(yi2)--(zi2)--(vi3);
\draw[very thick] (vi3)--(vi3');
\draw (vi3')--(wi3)--(xi3)--(yi3)--(zi3)--(vi1);
\end{tikzpicture} 
\quad
\begin{tikzpicture}[scale=.42] 
%\draw [color=gray!50] (0,0) grid (15,15);
\node[squareS] (vi1) at ({3+5*cos(0)},{4+5*sin(0)}) [label=right:\small $v_{i1}$] {};
\node[square] (vi1') at ({3+5*cos(20)},{4+5*sin(20)}) [label=right:\small $v_{i1}'$] {};
\node[squareS] (vi2) at ({3+5*cos(120)},{4+5*sin(120)}) [label=above:\small $v_{i2}$] {};
\node[square] (vi2') at ({3+5*cos(140)},{4+5*sin(140)}) [label=above:\small $v_{i2}'~$] {};
\node[squareS] (vi3) at ({3+5*cos(240)},{4+5*sin(240)}) [label=below:\small $v_{i3}$] {};
\node[square] (vi3') at ({3+5*cos(260)},{4+5*sin(260)}) [label=below:\small $v_{i3}'$] {};

\node[vertex] (wi1) at ({3+5*cos(40)},{4+5*sin(40)}) [label=above:\small $~w_{i1}$] {};
\node[vertexS] (xi1) at ({3+5*cos(60)},{4+5*sin(60)})  [label=above:\small $x_{i1}$] {};
\node[vertex] (yi1) at ({3+5*cos(80)},{4+5*sin(80)}) [label=above:\small $y_{i1}$] {};
\node[vertex] (zi1) at ({3+5*cos(100)},{4+5*sin(100)}) [label=above:\small $z_{i1}$] {};
\node[vertex] (wi2) at ({3+5*cos(160)},{4+5*sin(160)})  [label=left:\small $w_{i2}$] {};
\node[vertexS] (xi2) at ({3+5*cos(180)},{4+5*sin(180)}) [label=left:\small $x_{i2}$] {};
\node[vertex] (yi2) at ({3+5*cos(200)},{4+5*sin(200)}) [label=left:\small $y_{i2}$] {};
\node[vertex] (zi2) at ({3+5*cos(220)},{4+5*sin(220)}) [label=below:\small $z_{i2}$] {};
\node[vertex] (wi3) at ({3+5*cos(280)},{4+5*sin(280)}) [label=below:\small $w_{i3}$] {};
\node[vertexS] (xi3) at ({3+5*cos(300)},{4+5*sin(300)}) [label=below:\small $~x_{i3}$] {};
\node[vertex] (yi3) at ({3+5*cos(320)},{4+5*sin(320)}) [label=below:\small $~y_{i3}$] {};
\node[vertex] (zi3) at ({3+5*cos(340)},{4+5*sin(340)}) [label=right:\small $z_{i3}$] {};

\draw [very thick] (vi1)--(vi1');
\draw (vi1')--(wi1)--(xi1)--(yi1)--(zi1)--(vi2);
\draw[very thick] (vi2)--(vi2');
\draw (vi2')--(wi2)--(xi2)--(yi2)--(zi2)--(vi3);
\draw[very thick] (vi3)--(vi3');
\draw (vi3')--(wi3)--(xi3)--(yi3)--(zi3)--(vi1);
\end{tikzpicture} 
\caption{The variable gadget $G(v_i)$ in case $m=3$ (left) and an optimal cluster deletion set formed by the $2m$ black vertices (right).}\label{fig:CVDvi}
\end{center}
\end{figure}

The following properties of the variable gadget will be used: 
\begin{fact}\label{fact:vi}
$G(v_i)$ admits an optimal cluster deletion set of size $2m$. 
Any optimal cluster deletion set $S$ of $G(v_i)$ has the following properties:  
\begin{itemize}
\item[\em (a)] $S$ contains all or none of the variable vertices $v_{ij}$, $1\le j\le m$. The same holds for the variable vertices $v_{ij}'$, $1\le j\le m$;
\item[\em (b)] For any $1\le j\le m$, $v_{ij}$ and $v_{ij}'$ are not both in $S$. Moreover,  
%\item[\em (c)] For any $1\le j\le m$, 
each of $v_{ij}$ and $v_{ij}'$ has a neighbor outside $S$.
\end{itemize} 
\end{fact}
\begin{proof}
Observe that $G(v_i)$ can be partitioned into $m$ induced $P_3: v_{ij}v_{ij}'w_{ij}$ and $m$ induced $P_3:  x_{ij}y_{ij}z_{ij}$, $1\le j\le m$. 
Therefore, any cluster deletion set in $G(v_i)$ must contain a vertex of each $P_3$, hence has at least $2m$ vertices.  
Note also that $\{v_{ij}, x_{ij} \mid 1\le j\le m\}$, $\{v_{ij}', y_{ij} \mid 1\le j\le m\}$ and $\{w_{ij}, z_{ij} \mid 1\le j\le m\}$ are cluster deletion sets of size $2m$ (see also Fig.~\ref{fig:CVDvi} on the right hand). 

In particular, an optimal cluster deletion set must contain exactly one vertex of each $P_3$, showing (a) and (b). 
\qed 
\end{proof}

\textbf{Clause gadget.} For each clause $C_j$ consisting of variables $c_{j1}, c_{j2}$ and $c_{j3}$, let 
 $G(C_j)$ be the graph depicted on left hand side of Fig.~\ref{fig:CVDcj}; we call the six vertices labeled with $c_{jk}$ and $c_{jk}'$, $1\le k\le 3$, the \emph{clause vertices}. 

\begin{figure}[!ht]
\begin{center}
\tikzstyle{vertexS}=[draw,circle,inner sep=2pt,fill=black]
\tikzstyle{squareS}=[draw,rectangle,inner sep=2.5pt,fill=black]
\tikzstyle{subdS}=[draw,circle,inner sep=1.5pt,fill=black]
\tikzstyle{square}=[draw, rectangle,inner sep=2.5pt]
\tikzstyle{vertex}=[draw,circle,inner sep=2pt] 
\tikzstyle{subd}=[draw,circle,inner sep=1.5pt]

\begin{tikzpicture}[scale=.42] 
\node[square] (cj1) at ({3+5*cos(0)},{4+5*sin(0)}) [label=right:\small $c_{j1}$] {};
\node[square] (cj1') at ({3+5*cos(30)},{4+5*sin(30)}) [label=right:\small $c_{j1}'$] {};
\node[vertex] (x)  at ({3+2*cos(0)},{4+2*sin(0)}) {}; 
\node[vertex] (x1)  at ({3+3*cos(0)},{4+3*sin(0)}) {}; 
\node[vertex] (x2)  at ({3+4*cos(0)},{4+4*sin(0)}) {}; 
\node[vertex] (x')  at ({3+2*cos(30)},{4+2*sin(30)}) {}; 
\node[vertex] (x1')  at ({3+3*cos(30)},{4+3*sin(30)}) {}; 
\node[vertex] (x2')  at ({3+4*cos(30)},{4+4*sin(30)}) {}; 
\node[vertex] (p)  at ({3+2*cos(60)},{4+2*sin(60)}) {}; 
\node[vertex] (p')  at ({3+2*cos(90)},{4+2*sin(90)}) {}; 

\node[square] (cj2) at ({3+5*cos(120)},{4+5*sin(120)}) [label=left:\small $c_{j2}$] {};
\node[square] (cj2') at ({3+5*cos(150)},{4+5*sin(150)}) [label=left:\small $c_{j2}'$] {};
\node[vertex] (y)  at ({3+2*cos(120)},{4+2*sin(120)}) {}; 
\node[vertex] (y1)  at ({3+3*cos(120)},{4+3*sin(120)}) {}; 
\node[vertex] (y2)  at ({3+4*cos(120)},{4+4*sin(120)}) {}; 
\node[vertex] (y')  at ({3+2*cos(150)},{4+2*sin(150)}) {}; 
\node[vertex] (y1')  at ({3+3*cos(150)},{4+3*sin(150)}) {}; 
\node[vertex] (y2')  at ({3+4*cos(150)},{4+4*sin(150)}) {}; 
\node[vertex] (q)  at ({3+2*cos(180)},{4+2*sin(180)}) {}; 
\node[vertex] (q')  at ({3+2*cos(210)},{4+2*sin(210)}) {}; 

\node[square] (cj3) at ({3+5*cos(240)},{4+5*sin(240)}) [label=left:\small $c_{j3}$] {};
\node[square] (cj3') at ({3+5*cos(270)},{4+5*sin(270)}) [label=right:\small $c_{j3}'$] {};
\node[vertex] (z)  at ({3+2*cos(240)},{4+2*sin(240)}) {}; 
\node[vertex] (z1)  at ({3+3*cos(240)},{4+3*sin(240)}) {}; 
\node[vertex] (z2)  at ({3+4*cos(240)},{4+4*sin(240)}) {}; 
\node[vertex] (z')  at ({3+2*cos(270)},{4+2*sin(270)}) {}; 
\node[vertex] (z1')  at ({3+3*cos(270)},{4+3*sin(270)}) {}; 
\node[vertex] (z2')  at ({3+4*cos(270)},{4+4*sin(270)}) {}; 
\node[vertex] (r)  at ({3+2*cos(300)},{4+2*sin(300)}) {}; 
\node[vertex] (r')  at ({3+2*cos(330)},{4+2*sin(330)}) {}; 

\draw[very thick] (cj1)--(cj1'); \draw (cj1)--(x2)--(x1)--(x)--(x')--(x1')--(x2')--(cj1'); 
\draw[very thick] (cj2)--(cj2'); \draw (cj2)--(y2)--(y1)--(y)--(y')--(y1')--(y2')--(cj2'); 
\draw[very thick] (cj3)--(cj3'); \draw (cj3)--(z2)--(z1)--(z)--(z')--(z1')--(z2')--(cj3'); 
\draw (x')--(p)--(p')--(y); \draw (y')--(q)--(q')--(z); \draw (z')--(r)--(r')--(x);
\end{tikzpicture} 
\qquad
\begin{tikzpicture}[scale=.42] 
\node[squareS] (cj1) at ({3+5*cos(0)},{4+5*sin(0)}) [label=right:\small $c_{j1}$] {};
\node[square] (cj1') at ({3+5*cos(30)},{4+5*sin(30)}) [label=right:\small $c_{j1}'$] {};
\node[vertexS] (x)  at ({3+2*cos(0)},{4+2*sin(0)}) {}; 
\node[vertex] (x1)  at ({3+3*cos(0)},{4+3*sin(0)}) {}; 
\node[vertex] (x2)  at ({3+4*cos(0)},{4+4*sin(0)}) {}; 
\node[vertex] (x')  at ({3+2*cos(30)},{4+2*sin(30)}) {}; 
\node[vertexS] (x1')  at ({3+3*cos(30)},{4+3*sin(30)}) {}; 
\node[vertex] (x2')  at ({3+4*cos(30)},{4+4*sin(30)}) {}; 
\node[vertex] (p)  at ({3+2*cos(60)},{4+2*sin(60)}) {}; 
\node[vertexS] (p')  at ({3+2*cos(90)},{4+2*sin(90)}) {}; 

\node[squareS] (cj2) at ({3+5*cos(120)},{4+5*sin(120)}) [label=left:\small $c_{j2}$] {};
\node[square] (cj2') at ({3+5*cos(150)},{4+5*sin(150)}) [label=left:\small $c_{j2}'$] {};
\node[vertex] (y)  at ({3+2*cos(120)},{4+2*sin(120)}) {}; 
\node[vertexS] (y1)  at ({3+3*cos(120)},{4+3*sin(120)}) {}; 
\node[vertex] (y2)  at ({3+4*cos(120)},{4+4*sin(120)}) {}; 
\node[vertex] (y')  at ({3+2*cos(150)},{4+2*sin(150)}) {}; 
\node[vertexS] (y1')  at ({3+3*cos(150)},{4+3*sin(150)}) {}; 
\node[vertex] (y2')  at ({3+4*cos(150)},{4+4*sin(150)}) {}; 
\node[vertexS] (q)  at ({3+2*cos(180)},{4+2*sin(180)}) {}; 
\node[vertex] (q')  at ({3+2*cos(210)},{4+2*sin(210)}) {}; 

\node[square] (cj3) at ({3+5*cos(240)},{4+5*sin(240)}) [label=left:\small $c_{j3}$] {};
\node[squareS] (cj3') at ({3+5*cos(270)},{4+5*sin(270)}) [label=right:\small $c_{j3}'$] {};
\node[vertex] (z)  at ({3+2*cos(240)},{4+2*sin(240)}) {}; 
\node[vertexS] (z1)  at ({3+3*cos(240)},{4+3*sin(240)}) {}; 
\node[vertex] (z2)  at ({3+4*cos(240)},{4+4*sin(240)}) {}; 
\node[vertexS] (z')  at ({3+2*cos(270)},{4+2*sin(270)}) {}; 
\node[vertex] (z1')  at ({3+3*cos(270)},{4+3*sin(270)}) {}; 
\node[vertex] (z2')  at ({3+4*cos(270)},{4+4*sin(270)}) {}; 
\node[vertex] (r)  at ({3+2*cos(300)},{4+2*sin(300)}) {}; 
\node[vertex] (r')  at ({3+2*cos(330)},{4+2*sin(330)}) {}; 

\draw[very thick] (cj1)--(cj1'); \draw (cj1)--(x2)--(x1)--(x)--(x')--(x1')--(x2')--(cj1'); 
\draw[very thick] (cj2)--(cj2'); \draw (cj2)--(y2)--(y1)--(y)--(y')--(y1')--(y2')--(cj2'); 
\draw[very thick] (cj3)--(cj3'); \draw (cj3)--(z2)--(z1)--(z)--(z')--(z1')--(z2')--(cj3'); 
\draw (x')--(p)--(p')--(y); \draw (y')--(q)--(q')--(z); \draw (z')--(r)--(r')--(x);
\end{tikzpicture} 
\caption{The clause gadget $G(C_j)$ (left) and a cluster deletion set (black vertices) of size $11$ (right).}\label{fig:CVDcj}
\end{center}
\end{figure}
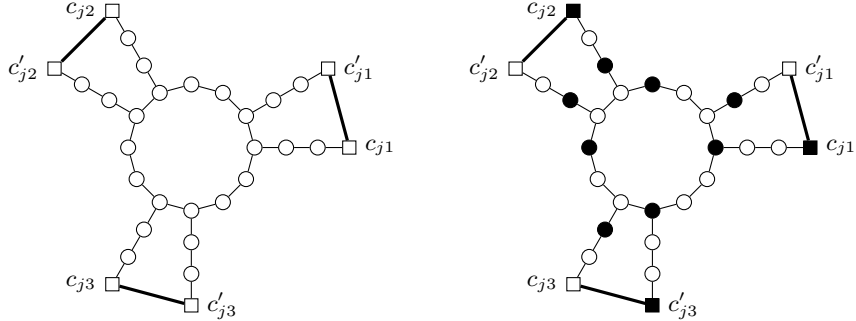

\begin{fact}\label{fact:cj1}
$G(C_j)$ admits an optimal cluster deletion set of size $11$. 
No optimal cluster deletion set of $G(C_j)$ contains all $c_{j1}$, $c_{j2}$ and $c_{j3}$, or all $c_{j1}'$, $c_{j2}'$ and $c_{j3}'$.
\end{fact}
\begin{proof}
For $k\in\{1,2,3\}$, write $x_{jk}, x_{jk}'$ for the two adjacent degree-$3$ vertices in $G(C_j)$ which together with $c_{jk}$ and $c_{jk}'$ belong to a $8$-cycle. 
Let $H$ be the $12$-cycle containing all $x_{jk}$ and $x_{jk}'$. Note that $G(C_j)$ minus $H$ consists of three connected components each of which is a $6$-vertex path with midpoints $c_{jk}$ and $c_{jk}'$, $1\le k\le 3$. 
Note also that any cluster deletion set of the $12$-cycle $H$ has at least~$4$ vertices, and any cluster deletion set of a $6$-vertex path has at least~$2$ vertices. 

Consider a cluster deletion set $S$ in $G(C_j)$. If $S$ contains at least $5$ verticies from $H$ then $S$ has at least~$11$ vertices. Let us assume that $S$ contains exactly~$4$ vertices from $H$. Then it can be verify that, for some $1\le k\le 3$, 
both $x_{jk}, x_{jk}'$ are outside $S$, and therefore, $S$ contains at least $3$ vertices from the $6$-path with midpoints $c_{jk}, c_{jk}'$. Thus again, $S$ contains at least~$11$ vertices. The first part follows now by noting that $G(C_j)$ admits a cluster deletion set of size~$11$ as depicted in Fig.~\ref{fig:CVDcj}.
   
Moreover, observe that $G(C_j)$ minus $c_{j1}, c_{j2}$ and $c_{j3}$  has a partition into~$9$ disjoint $P_3$. 
Hence no optimal cluster deletion set in $G(C_j)$ can contain all $c_{j1}, c_{j2}$ and $c_{j3}$, for otherwise it would contain at least~$12$ vertices. Similarly for $c_{j1}', c_{j2}'$ and $c_{j3}'$.      
\qed
\end{proof}

Finally, the graph $G$ is obtained by connecting the variable and clause gadgets as follows: if variable $v_i$ appears in clause $C_j$, i.e., $v_i=c_{jk}$ for some $k\in\{1,2,3\}$, then 

\begin{itemize}
\item connect the variable vertex $v_{ij}$ in $G(v_i)$ and the clause vertex $c_{jk}$ in $G(C_j)$ by an edge; $v_{ij}$ is the \emph{corresponding variable vertex} of the clause vertex $c_{jk}$, and
\item connect the variable vertex $v_{ij}'$ in $G(v_i)$ and the clause vertex $c_{jk}'$ in $G(C_j)$ by an edge; $v_{ij}'$ is the \emph{corresponding variable vertex} of the clause vertex $c_{jk}'$.
\end{itemize} 

\medskip
\begin{fact}\label{fact:bip}
$G$ has maximum degree~$3$ and is bipartite.
\end{fact}
\begin{proof}
It follows from the construction that $G$ has maximum degree~$3$. 
To see that $G$ is bipartite, note that the bipartite graph forming by all variable gadgets $G(v_i)$ has a bipartition $(A,B)$ such that all $v_{ij}$ are in $A$ and all $v_{ij}'$ are in $B$, and the bipartite graph forming by all clause gadgets $G(C_j)$ has a bipartition $(C,D)$ such that all $c_{j1}, c_{j2}, c_{j3}$ are in $C$ and all $c_{j1}', c_{j2}', c_{j3}'$ are in~$D$. Hence, by construction, $(A\cup D, B\cup C)$ is a bipartition of $G$.\qed
\end{proof}

The following fact will be important for later discussion on cluster deletion sets in the clause gadget.

\begin{fact}\label{fact:cj2}
Let $S$ be a cluster deletion set in $G$ such that $S$ contains exactly~$2m$ vertices from each $G(v_i)$. 
Then, for any $1\le j\le m$ and $1\le k\le 3$, $c_{jk}\in S$ or~$c_{jk}'\in S$. 
Moreover, if the clause vertex $c\in\{c_{jk},c_{jk}'\}$ is not in $S$ then the corresponding variable vertex of $c$ is in $S$;   
\end{fact} 
\begin{proof}
By Fact~\ref{fact:vi}, the restriction of $S$ on each $G(v_i)$ is an optimal cluster deletion set in $G(v_i)$. 
Hence, by Fact~\ref{fact:vi}~(b), for every $1\le j\le m$, some of the variable vertices $v_{ij}, v_{ij}'$ is not in $S$. 
Thus, if both $c_{jk}$ and $c_{jk}'$ are not in $S$ then with their corresponding variable vertices we have an induced $P_3$ outside $S$, a contradiction.

Moreover, if some $c\in\{c_{jk},c_{jk}'\}$ is not in $S$ then the corresponding variable vertex $v$ %\in\{v_{rj}, v_{rj}'\}$ 
is in~$S$. For, otherwise $c, v$ and a neighbor of $v$ outside $S$ (which exists by Fact~\ref{fact:vi}~(b)) would induce a $P_3$ outside $S$.
\qed 
\end{proof}

%Set $k=2nm+11m$. %We now show that $F\in \text{\PNAE3SAT}$ if and only if $(G,k)\in\text{\CVD}$.  
Now suppose that $G$ has a cluster deletion set $S$ with at most $2nm+11m$ vertices. 
Then by Facts~\ref{fact:vi} and~\ref{fact:cj1}, $S$ has exactly $2nm+11m$ vertices, and~$S$ contains exactly $2m$ vertices from each $G(v_i)$ and exactly $11$ vertices from each $G(C_j)$. 
%Moreover, by Fact~\ref{obs:cj} and Fact~\ref{fact:cj1}~(a), we have: For any $1\le j\le m$, if both $c_{j1}$ and $c_{j1}'$ are in $S$ then exactly one of $c_{j2}, c_{j2}'$ is in $S$ and exactly one of $c_{j3}, c_{j3}'$ is in $S$. Analogously in case both $c_{j2}$ and $c_{j2}'$ are in $S$, respectively, both $c_{j2}$ and $c_{j2}'$ are in $S$. 

Consider the truth assignment in which a variable $v_i$ is true if all its associated variable vertices $v_{ij}$, $1\le j\le m$, are in $S$. Note that by Fact~\ref{fact:vi}~(a), this assignment is well-defined. For each $G(C_j)$, it follows from Fact~\ref{fact:cj1} that some of $c_{j1}, c_{j2}, c_{j3}$ is not in $S$ and some of $c_{j1}', c_{j2}', c_{j3}'$ is not in $S$. 
Let $c_{j1}\not\in S$, say. By Fact~\ref{fact:cj2}, $c_{j1}'\in S$. 
Hence $c_{j2}'\not\in S$ or $c_{j3}'\not\in S$, and again by Fact~\ref{fact:cj2}, $c_{j2}\in S$ or $c_{j3}\in S$.   
Let $c_{j2}\in S$; the case $c_{j3}\in S$ is similar. 
Let $v_{rj}$ and $v_{sj}$ be the corresponding variable vertices of $c_{j1}$ and $c_{j2}$, respectively. 
Then by Fact~\ref{fact:cj2} again, $v_{rj}\in S$ and $v_{sj}'\in S$. 
By Fact~\ref{fact:vi}~(b), $v_{sj}\notin S$.   
That is, the clause $C_j$ contains a true variable $v_r$ and a false variable $v_s$. 
Thus, if $G$ admits a cluster deletion set $S$ with at most $2nm+11m$ vertices then $F$ has a nae assignment.  

Conversely, suppose that there is a nae assignment for~$F$. Then a cluster deletion set $S$ of size $2nm+11m$ for $G$ is as follows. 
For each variable $v_i$, $1\le i\le n$: 
\begin{itemize}
\item If variable $v_i$ is true, then put all $2m$ vertices $v_{ij}, x_{ij}$, $1\le j\le m$, into~$S$. 
\item If variable $v_i$ is false, then put all $2m$ vertices $v_{ij}', y_{ij}$, $1\le j\le m$, into~$S$.  
\end{itemize}

For each clause $C_j$, $1\le j\le m$, let $v_{rj}$, $v_{sj}$ and $v_{tj}$ be the corresponding variable vertices of $c_{j1}$, $c_{j2}$ and $c_{j3}$, respectively. Extend $S$ to~$11$ vertices of~$G(C_j)$ as follows:  
\begin{itemize}
\item If $C_j$ has one true variable and two false variables, say $v_r$ is true, $v_s$ and $v_t$ are false, then put the clause vertices $c_{j1}'$, $c_{j2}$ and $c_{j3}$ into $S$ and extend $S$ to another~$8$ vertices as indicated in Fig.~\ref{fig:extendGCj} on the left hand.
\item If $C_j$ has two true variables and one false variable, say $v_r$ and $v_s$ are true, $v_t$ is false, then put the clause vertices $c_{j1}'$, $c_{j2}'$ and $c_{j3}$ into $S$ and extend $S$ to another~$8$ vertices as indicated in Fig.~\ref{fig:extendGCj} on the right hand.
\end{itemize}

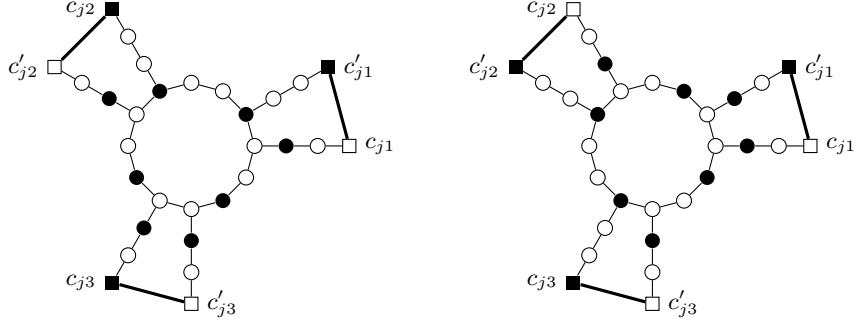
\begin{figure}[!ht]
\begin{center}
\tikzstyle{vertexS}=[circle,inner sep=2pt,fill=black]
\tikzstyle{squareS}=[rectangle,inner sep=2.5pt,fill=black]
\tikzstyle{square}=[draw, rectangle,inner sep=2.5pt]
\tikzstyle{vertex}=[draw,circle,inner sep=2pt] 
\tikzstyle{subd}=[draw,circle,inner sep=1.5pt]

\begin{tikzpicture}[scale=.42] 
\node[square] (cj1) at ({3+5*cos(0)},{4+5*sin(0)}) [label=right:\small $c_{j1}$] {};
\node[squareS] (cj1') at ({3+5*cos(30)},{4+5*sin(30)}) [label=right:\small $c_{j1}'$] {};
\node[vertex] (x)  at ({3+2*cos(0)},{4+2*sin(0)}) {}; 
\node[vertexS] (x1)  at ({3+3*cos(0)},{4+3*sin(0)}) {}; 
\node[vertex] (x2)  at ({3+4*cos(0)},{4+4*sin(0)}) {}; 
\node[vertexS] (x')  at ({3+2*cos(30)},{4+2*sin(30)}) {}; 
\node[vertex] (x1')  at ({3+3*cos(30)},{4+3*sin(30)}) {}; 
\node[vertex] (x2')  at ({3+4*cos(30)},{4+4*sin(30)}) {}; 
\node[vertex] (p)  at ({3+2*cos(60)},{4+2*sin(60)}) {}; 
\node[vertex] (p')  at ({3+2*cos(90)},{4+2*sin(90)}) {}; 

\node[squareS] (cj2) at ({3+5*cos(120)},{4+5*sin(120)}) [label=left:\small $c_{j2}$] {};
\node[square] (cj2') at ({3+5*cos(150)},{4+5*sin(150)}) [label=left:\small $c_{j2}'$] {};
\node[vertexS] (y)  at ({3+2*cos(120)},{4+2*sin(120)}) {}; 
\node[vertex] (y1)  at ({3+3*cos(120)},{4+3*sin(120)}) {}; 
\node[vertex] (y2)  at ({3+4*cos(120)},{4+4*sin(120)}) {}; 
\node[vertex] (y')  at ({3+2*cos(150)},{4+2*sin(150)}) {}; 
\node[vertexS] (y1')  at ({3+3*cos(150)},{4+3*sin(150)}) {}; 
\node[vertex] (y2')  at ({3+4*cos(150)},{4+4*sin(150)}) {}; 
\node[vertex] (q)  at ({3+2*cos(180)},{4+2*sin(180)}) {}; 
\node[vertexS] (q')  at ({3+2*cos(210)},{4+2*sin(210)}) {}; 

\node[squareS] (cj3) at ({3+5*cos(240)},{4+5*sin(240)}) [label=left:\small $c_{j3}$] {};
\node[square] (cj3') at ({3+5*cos(270)},{4+5*sin(270)}) [label=right:\small $c_{j3}'$] {};
\node[vertex] (z)  at ({3+2*cos(240)},{4+2*sin(240)}) {}; 
\node[vertexS] (z1)  at ({3+3*cos(240)},{4+3*sin(240)}) {}; 
\node[vertex] (z2)  at ({3+4*cos(240)},{4+4*sin(240)}) {}; 
\node[vertex] (z')  at ({3+2*cos(270)},{4+2*sin(270)}) {}; 
\node[vertexS] (z1')  at ({3+3*cos(270)},{4+3*sin(270)}) {}; 
\node[vertex] (z2')  at ({3+4*cos(270)},{4+4*sin(270)}) {}; 
\node[vertexS] (r)  at ({3+2*cos(300)},{4+2*sin(300)}) {}; 
\node[vertex] (r')  at ({3+2*cos(330)},{4+2*sin(330)}) {}; 

\draw[very thick] (cj1)--(cj1'); \draw (cj1)--(x2)--(x1)--(x)--(x')--(x1')--(x2')--(cj1'); 
\draw[very thick] (cj2)--(cj2'); \draw (cj2)--(y2)--(y1)--(y)--(y')--(y1')--(y2')--(cj2'); 
\draw[very thick] (cj3)--(cj3'); \draw (cj3)--(z2)--(z1)--(z)--(z')--(z1')--(z2')--(cj3'); 
\draw (x')--(p)--(p')--(y); \draw (y')--(q)--(q')--(z); \draw (z')--(r)--(r')--(x);
\end{tikzpicture} 
\qquad
\begin{tikzpicture}[scale=.42] 
\node[square] (cj1) at ({3+5*cos(0)},{4+5*sin(0)}) [label=right:\small $c_{j1}$] {};
\node[squareS] (cj1') at ({3+5*cos(30)},{4+5*sin(30)}) [label=right:\small $c_{j1}'$] {};
\node[vertex] (x)  at ({3+2*cos(0)},{4+2*sin(0)}) {}; 
\node[vertexS] (x1)  at ({3+3*cos(0)},{4+3*sin(0)}) {}; 
\node[vertex] (x2)  at ({3+4*cos(0)},{4+4*sin(0)}) {}; 
\node[vertex] (x')  at ({3+2*cos(30)},{4+2*sin(30)}) {}; 
\node[vertexS] (x1')  at ({3+3*cos(30)},{4+3*sin(30)}) {}; 
\node[vertex] (x2')  at ({3+4*cos(30)},{4+4*sin(30)}) {}; 
\node[vertexS] (p)  at ({3+2*cos(60)},{4+2*sin(60)}) {}; 
\node[vertex] (p')  at ({3+2*cos(90)},{4+2*sin(90)}) {}; 

\node[square] (cj2) at ({3+5*cos(120)},{4+5*sin(120)}) [label=left:\small $c_{j2}$] {};
\node[squareS] (cj2') at ({3+5*cos(150)},{4+5*sin(150)}) [label=left:\small $c_{j2}'$] {};
\node[vertex] (y)  at ({3+2*cos(120)},{4+2*sin(120)}) {}; 
\node[vertexS] (y1)  at ({3+3*cos(120)},{4+3*sin(120)}) {}; 
\node[vertex] (y2)  at ({3+4*cos(120)},{4+4*sin(120)}) {}; 
\node[vertexS] (y')  at ({3+2*cos(150)},{4+2*sin(150)}) {}; 
\node[vertex] (y1')  at ({3+3*cos(150)},{4+3*sin(150)}) {}; 
\node[vertex] (y2')  at ({3+4*cos(150)},{4+4*sin(150)}) {}; 
\node[vertex] (q)  at ({3+2*cos(180)},{4+2*sin(180)}) {}; 
\node[vertex] (q')  at ({3+2*cos(210)},{4+2*sin(210)}) {}; 

\node[squareS] (cj3) at ({3+5*cos(240)},{4+5*sin(240)}) [label=left:\small $c_{j3}$] {};
\node[square] (cj3') at ({3+5*cos(270)},{4+5*sin(270)}) [label=right:\small $c_{j3}'$] {};
\node[vertexS] (z)  at ({3+2*cos(240)},{4+2*sin(240)}) {}; 
\node[vertex] (z1)  at ({3+3*cos(240)},{4+3*sin(240)}) {}; 
\node[vertex] (z2)  at ({3+4*cos(240)},{4+4*sin(240)}) {}; 
\node[vertex] (z')  at ({3+2*cos(270)},{4+2*sin(270)}) {}; 
\node[vertexS] (z1')  at ({3+3*cos(270)},{4+3*sin(270)}) {}; 
\node[vertex] (z2')  at ({3+4*cos(270)},{4+4*sin(270)}) {}; 
\node[vertex] (r)  at ({3+2*cos(300)},{4+2*sin(300)}) {}; 
\node[vertexS] (r')  at ({3+2*cos(330)},{4+2*sin(330)}) {}; 

\draw[very thick] (cj1)--(cj1'); \draw (cj1)--(x2)--(x1)--(x)--(x')--(x1')--(x2')--(cj1'); 
\draw[very thick] (cj2)--(cj2'); \draw (cj2)--(y2)--(y1)--(y)--(y')--(y1')--(y2')--(cj2'); 
\draw[very thick] (cj3)--(cj3'); \draw (cj3)--(z2)--(z1)--(z)--(z')--(z1')--(z2')--(cj3'); 
\draw (x')--(p)--(p')--(y); \draw (y')--(q)--(q')--(z); \draw (z')--(r)--(r')--(x);
\end{tikzpicture} 
\caption{An optimal cluster deletion set (black vertices) in $G(C_j)$; left hand: one true ($c_{j1}$) two false ($c_{j2}$ and $c_{j3}$) variables, right hand: two true ($c_{j1}$ and $c_{j2}$) one false ($c_{j3}$) variables.}\label{fig:extendGCj}
\end{center}
\end{figure}

By construction, $S$ has exactly $n\times 2m+m\times 11$ vertices and, for every $i$ and $j$, $G(v_i)-S$ and $G(C_j)-S$ are $P_3$-free. Therefore, since a clause vertex is outside $S$ if and only if the corresponding variable vertex is in $S$, $G-S$ is $P_3$-free. Thus, if $F$ can be satisfied by a nae assignment then $G$ has a cluster deletion set of size at most $2nm+11m$.

In summary, we obtain:
\begin{theorem}\label{thm:bip+deg-3}
\CVD\ is $\NP$-complete even when restricted to bipartite graphs of maximum degree~$3$. 
%For any $d\ge 3$, \dclawVD{d} is $\NP$-complete even when restricted to bipartite graphs of maximum degree~$d$.
\end{theorem}

\subsubsection{$d$-Claw Vertex Deletion is hard in bipartite graphs of maximum degree~$d$}
%-----------------------------------------------------------------------------------------

We first give a polynomial-time reduction from \PNAE3SAT\ to \dclawVD{3}. The case $d>3$ will be easily derived from this case.

Let $F=C_1\land C_2\land\cdots\land C_m$ over $n$ variables $x_1, x_2, \ldots, x_n$, in which each clause $C_j$ consists of three distinct variables. 
We will construct a subcubic bipartite graph $G$ such that $G$ has a claw deletion set of size at most $2mn+16m$ if and only if $F$ admits a nae assignment. The graph $G$ contains a gadget $G(v_i)$ for each variable $v_i$ and a gadget $G(C_j)$ for each clause $C_j$. 

For any $1\le i\le n$ and $1\le j\le m$, we first consider an auxiliary $8$-vertex graph $H_{ij}$ depicted on the left hand side of Fig.~\ref{fig:Hij} which will be useful in building our variable gadget.
 
\begin{figure}[!ht]
\begin{center}
\tikzstyle{vertexS}=[draw,circle,inner sep=2pt,fill=black]
\tikzstyle{squareS}=[draw,rectangle,inner sep=2.5pt,fill=black]
\tikzstyle{subdS}=[draw,circle,inner sep=1.5pt,fill=black]
\tikzstyle{square}=[draw, rectangle,inner sep=2.5pt]
\tikzstyle{vertex}=[draw,circle,inner sep=2pt] 
\tikzstyle{subd}=[draw,circle,inner sep=1.5pt]
\begin{tikzpicture}[scale=.345] 
\node[square] (vij) at (1,1) [label=below:\small $v_{ij}$] {};
\node[square] (vij') at (7,7) [label=above:\small $v_{ij}'$] {};

\node[vertex] (aij1) at (4,1) [label=below:\small $a_{ij}^1$] {};
\node[vertex] (aij2) at (7,1) [label=below:\small $a_{ij}^2$] {}; % 
\node[vertex] (aij3) at (10,1) [label=below:\small $a_{ij}^3$] {}; % 
\node[vertex] (bij1) at (4,4) [label=left:\small $b_{ij}^1$] {}; % 
\node[vertex] (bij2) at (7,4) [label=above:\small $b_{ij}^2$] {}; % 
\node[vertex] (bij3) at (10,4) [label=right:\small $b_{ij}^3$] {}; % 

\draw (vij)--(aij1)--(aij2)--(aij3)--(bij3)--(bij2)--(bij1)--(aij1);
\draw (aij2)--(bij2);
\draw (bij1)--(vij')--(bij3);
\end{tikzpicture} 
\qquad
\begin{tikzpicture}[scale=.345] 
\node[squareS] (vij) at (1,1) [label=below:\small $v_{ij}$] {};
\node[square] (vij') at (7,7) [label=above:\small $v_{ij}'$] {};

\node[vertex] (aij1) at (4,1) [label=below:\small $a_{ij}^1$] {};
\node[vertex] (aij2) at (7,1) [label=below:\small $a_{ij}^2$] {}; % 
\node[vertex] (aij3) at (10,1) [label=below:\small $a_{ij}^3$] {}; % 
\node[vertex] (bij1) at (4,4) [label=left:\small $b_{ij}^1$] {}; % 
\node[vertexS] (bij2) at (7,4) [label=above:\small $b_{ij}^2$] {}; % 
\node[vertex] (bij3) at (10,4) [label=right:\small $b_{ij}^3$] {}; % 

\draw (vij)--(aij1)--(aij2)--(aij3)--(bij3)--(bij2)--(bij1)--(aij1);
\draw (aij2)--(bij2);
\draw (bij1)--(vij')--(bij3);
\end{tikzpicture} 
\caption{Left: the auxiliary graph $H_{ij}$. Right: the unique optimal claw deletion set (back vertices) containing $v_{ij}$.}\label{fig:Hij}
\end{center}
\end{figure}
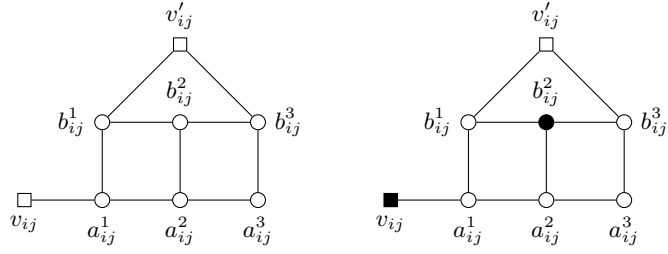

We need the following property of $H_{ij}$ which can be immediately verified:
\begin{fact}\label{fact:Hij}
Any optimal claw deletion set in $H_{ij}$ contains exactly~$2$ vertices. 
Moreover, if $S$ is an optimal claw deletion set containing $v_{ij}$ then $S=\{v_{ij}, b_{ij}^2\}$. 
\end{fact}
%\begin{proof}
%Since $H_{ij}$ has a partition into $2$ disjoint claws, any claw deletion must contain at least~$2$ vertices. 
%The second part can be seen by inspection.\qed
%\end{proof}

We now are ready to describe the variable gadget.

\smallskip\noindent
\textbf{Variable gadget.} For each variable $v_i$, take first $m$ auxiliary graphs $H_{i1}$, $H_{i2}$, \ldots, $H_{im}$, one for each clause. 
Then, for all $1\le j<m$, connect the vertex $a_{ij}^3$ in $H_{ij}$ with the vertex $v_{i(j+1)}$ in $H_{i(j+1)}$ by an edge. Last, connect the vertex $a_{im}^3$ in $H_{im}$ with the vertex $v_{i1}$ in $H_{i1}$ by an edge. 
The obtained graph is denoted by $G(v_i)$. The case $m=3$ is shown in Fig.~\ref{fig:clawVDvi}. The $2m$ vertices $v_{ij}$ and $v_{ij}'$, $1\le j\le m$, in $G(v_i)$ are the \emph{variable vertices}. 

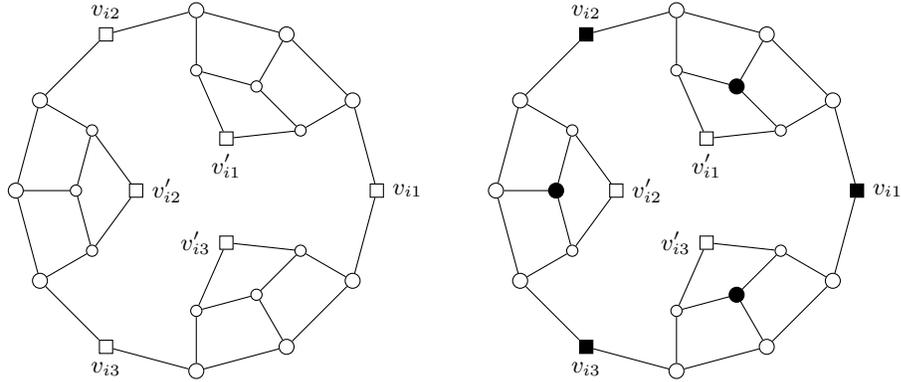
\begin{figure}[!ht]
\begin{center}
\tikzstyle{vertexS}=[draw,circle,inner sep=2pt,fill=black]
\tikzstyle{squareS}=[draw,rectangle,inner sep=2.5pt,fill=black]
\tikzstyle{subdS}=[draw,circle,inner sep=1.5pt,fill=black]
\tikzstyle{square}=[draw, rectangle,inner sep=2.5pt]
\tikzstyle{vertex}=[draw,circle,inner sep=2pt] 
\tikzstyle{subd}=[draw,circle,inner sep=1.5pt]
\begin{tikzpicture}[scale=.4] 
%\draw [color=gray!50] (0,0) grid (15,15);

\node[square] (vi1) at ({3+6*cos(0)},{4+6*sin(0)}) [label=right:\small $v_{i1}$] {};
\node[square] (vi1') at ({3+2*cos(60)},{4+2*sin(60)}) [label=below:\small $v_{i1}'$] {};
\node[square] (vi2) at ({3+6*cos(120)},{4+6*sin(120)}) [label=above:\small $v_{i2}$] {};
\node[square] (vi2') at ({3+2*cos(180)},{4+2*sin(180)}) [label=right:\small $v_{i2}'$] {};
\node[square] (vi3) at ({3+6*cos(240)},{4+6*sin(240)}) [label=below:\small $v_{i3}$] {};
\node[square] (vi3') at ({3+2*cos(300)},{4+2*sin(300)}) [label=left:\small $v_{i3}'$] {};

\foreach \angle / \i in {30/1, 60/2, 90/3, 150/4, 180/5, 
        210/6, 270/7, 300/8, 330/9}  
        {
    \node[vertex] (a\i) at ({3+6*cos(\angle)},{4+6*sin(\angle)}) {};
    \node[subd] (b\i) at ({3+4*cos(\angle)},{4+4*sin(\angle)}) {};
} 

\draw (vi1)--(a1)--(a2)--(a3)--(b3)--(b2)--(b1)--(a1);
\draw (b1)--(vi1')--(b3); \draw (a2)--(b2);
\draw (vi2)--(a4)--(a5)--(a6)--(b6)--(b5)--(b4)--(a4);
\draw (b4)--(vi2')--(b6); \draw (a5)--(b5);
\draw (vi3)--(a7)--(a8)--(a9)--(b9)--(b8)--(b7)--(a7);
\draw (b9)--(vi3')--(b7); \draw (a8)--(b8);

\draw (a3)--(vi2); \draw (a6)--(vi3); \draw (a9)--(vi1);
\end{tikzpicture} 
\qquad
\begin{tikzpicture}[scale=.4] 
%\draw [color=gray!50] (0,0) grid (15,15);

\node[squareS] (vi1) at ({3+6*cos(0)},{4+6*sin(0)}) [label=right:\small $v_{i1}$] {};
\node[square] (vi1') at ({3+2*cos(60)},{4+2*sin(60)}) [label=below:\small $v_{i1}'$] {};
\node[squareS] (vi2) at ({3+6*cos(120)},{4+6*sin(120)}) [label=above:\small $v_{i2}$] {};
\node[square] (vi2') at ({3+2*cos(180)},{4+2*sin(180)}) [label=right:\small $v_{i2}'$] {};
\node[squareS] (vi3) at ({3+6*cos(240)},{4+6*sin(240)}) [label=below:\small $v_{i3}$] {};
\node[square] (vi3') at ({3+2*cos(300)},{4+2*sin(300)}) [label=left:\small $v_{i3}'$] {};

\node[vertexS]  at ({3+4*cos(60)},{4+4*sin(60)}) {}; %bi1^2
\node[vertexS]  at ({3+4*cos(180)},{4+4*sin(180)}) {}; %bi2^2
\node[vertexS]  at ({3+4*cos(300)},{4+4*sin(300)}) {}; %bi3^2

\foreach \angle / \i in {30/1, 60/2, 90/3, 150/4, 180/5, 
        210/6, 270/7, 300/8, 330/9}  
        {
    \node[vertex] (a\i) at ({3+6*cos(\angle)},{4+6*sin(\angle)}) {};
    \node[subd] (b\i) at ({3+4*cos(\angle)},{4+4*sin(\angle)}) {};
} 

\draw (vi1)--(a1)--(a2)--(a3)--(b3)--(b2)--(b1)--(a1);
\draw (b1)--(vi1')--(b3); \draw (a2)--(b2);
\draw (vi2)--(a4)--(a5)--(a6)--(b6)--(b5)--(b4)--(a4);
\draw (b4)--(vi2')--(b6); \draw (a5)--(b5);
\draw (vi3)--(a7)--(a8)--(a9)--(b9)--(b8)--(b7)--(a7);
\draw (b9)--(vi3')--(b7); \draw (a8)--(b8);

\draw (a3)--(vi2); \draw (a6)--(vi3); \draw (a9)--(vi1);
\end{tikzpicture} 
\caption{The variable gadget $G(v_i)$ in case $m=3$ (left) and an optimal claw deletion set formed by the $2m$ black vertices (right).}
\label{fig:clawVDvi}
\end{center}
\end{figure}

The following properties of the variable gadget will be used: 
\begin{fact}\label{fact:claw-vi1}
$G(v_i)$ admits an optimal cluster deletion set of size $2m$. 
Moreover, if $S$ is an optimal claw deletion set in $G(v_i)$ then the restriction of $S$ on $H_{ij}$ is an optimal claw deletion set in $H_{ij}$.
\end{fact}
\begin{proof}
$G(v_i)$ consists of $m$ disjoint $H_{ij}$, hence, by Fact~\ref{fact:Hij}, any claw deletion set in $G(v_i)$ must contain at least~$2$ vertices from each $H_{ij}$.  
Observe that the union of all $\{v_{ij}, b_{ij}^2\}$, $1\le j\le m$, is a claw deletion set of $G(v_i)$ of size $2m$. 

The second part follows from the first part and Fact~\ref{fact:Hij}. \qed
\end{proof}

\begin{fact}\label{fact:claw-vi2}
Let $S$ be an optimal claw deletion in $G(v_i)$. Then: 
\begin{itemize}
\item[\em (a)] If some $v_{ij}$ is contained in $S$ then $S=\bigcup_{1\le j\le m} \{v_{ij}, b_{ij}^2\}$; 
in particular, all~$v_{ij}$ are in $S$, and all $v_{ij}'$ and their neighbors are outside~$S$.
\item[\em (b)] If some $v_{ij}$ is not contained in $S$ then all $v_{ij}$ are outside $S$.
%\item[\em (c)] For any $j$, if $v_{ij}'\in S$ then $v_{ij}\not\in S$.  (Wozu?)
\end{itemize}
\end{fact}
\begin{proof}
(a): By Fact~\ref{fact:claw-vi1}, the restriction of $S$ on $H_{ij}$ is an optimal claw deletion set in $H_{ij}$, which is $\{v_{ij}, b_{ij}^2\}$ by Fact~\ref{fact:Hij}. 
In particular, $v_{ij}'\not\in S$, and the $P_3$: $a_{ij}^2a_{ij}^3b_{ij}^3$ is outside $S$ as well, implying $v_{i(j+1)}\in S$. 
Similarly, by Fact~\ref{fact:claw-vi1} again, the restriction of $S$ on $H_{i(j+1)}$ is an optimal claw deletion set in $H_{i(j+1)}$, which is $\{v_{i(j+1)}, b_{i(j+1)}^2\}$ by Fact~\ref{fact:Hij}. 
Apply Fact~\ref{fact:Hij} for $H_{i(j+1)}$ we have $v_{i(j+1)}'\not\in S$, $v_{i(j+2)}\in S $ and so on. 

(b) follows from part (a). %and (c) follow from part (a).
\qed 
\end{proof}

Before presenting the clause gadget, we need other auxiliary graphs $A_{jk}$, $A'_{jk}$ for any $1\le j\le m$ and $1\le k\le 3$ depicted in Fig.~\ref{fig:Ajk}.

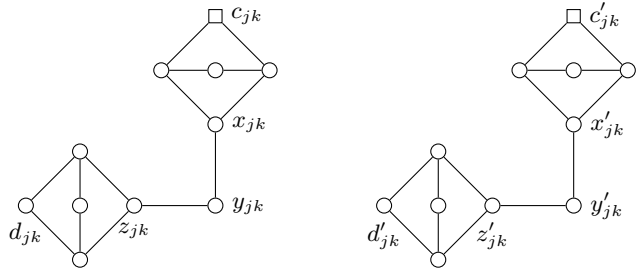
\begin{figure}[!ht]
\tikzstyle{vertexS}=[draw,circle,inner sep=2pt,fill=black]
\tikzstyle{squareS}=[draw,rectangle,inner sep=2.5pt,fill=black]
\tikzstyle{subdS}=[draw,circle,inner sep=1.5pt,fill=black]
\tikzstyle{square}=[draw, rectangle,inner sep=2.5pt]
\tikzstyle{vertex}=[draw,circle,inner sep=2pt] 
\tikzstyle{subd}=[draw,circle,inner sep=1.5pt]
\begin{center}
\begin{tikzpicture}[scale=.36] 
\node[square] (cjk) at (8,10) [label=right:\small $c_{jk}$] {};
\node[vertex] (p1) at (6,8) {};
\node[vertex] (p2) at (8,8) {};
\node[vertex] (p3) at (10,8) {};
\node[vertex] (xjk) at (8,6) [label=right:\small $x_{jk}$] {};  
\node[vertex] (yjk) at (8,3) [label=right:\small $y_{jk}$] {}; % 
\node[vertex] (zjk) at (5,3) [label=below:\small $z_{jk}$] {};
\node[vertex] (q1) at (3,1) {};
\node[vertex] (q2) at (3,3) {};
\node[vertex] (q3) at (3,5) {};
\node[vertex] (djk) at (1,3) [label=below:\small $d_{jk}$] {};

\draw (cjk)--(p1)--(xjk)--(yjk)--(zjk)--(q1)-- (djk)--(q3)--(zjk);
\draw (cjk)--(p3)--(xjk);
\draw (p1)--(p2)--(p3);
\draw (q1)--(q2)--(q3);

\end{tikzpicture} 
\qquad\quad
\begin{tikzpicture}[scale=.36] 
\node[square] (cjk) at (8,10) [label=right:\small $c_{jk}'$] {};
\node[vertex] (p1) at (6,8) {};
\node[vertex] (p2) at (8,8) {};
\node[vertex] (p3) at (10,8) {};
\node[vertex] (xjk) at (8,6) [label=right:\small $x_{jk}'$] {};  
\node[vertex] (yjk) at (8,3) [label=right:\small $y_{jk}'$] {}; % 
\node[vertex] (zjk) at (5,3) [label=below:\small $z_{jk}'$] {};
\node[vertex] (q1) at (3,1) {};
\node[vertex] (q2) at (3,3) {};
\node[vertex] (q3) at (3,5) {};
\node[vertex] (djk) at (1,3) [label=below:\small $d_{jk}'$] {};

\draw (cjk)--(p1)--(xjk)--(yjk)--(zjk)--(q1)-- (djk)--(q3)--(zjk);
\draw (cjk)--(p3)--(xjk);
\draw (p1)--(p2)--(p3);
\draw (q1)--(q2)--(q3);
\end{tikzpicture} 
\caption{The auxiliary graphs $A_{jk}$ (left) and $A'_{jk}$ (right).}\label{fig:Ajk}
\end{center}
\end{figure}

Observing that any optimal claw deletion set in the complete bipartite graph $K_{2,3}$ consists of one degree-$2$ vertex and that $A_{jk}$ contains  two disjoint $K_{2,3}$, the following fact follows immediately:
\begin{fact}\label{fact:Ajk}
Any claw deletion set in $A_{jk}$ has at least~$2$ vertices; $\{x_{jk}, z_{jk}\}$ is the only optimal claw deletion set of size~$2$. 
%In particular, any claw deletion set not containing $a_{jk}$ or $c_{jk}$ contains at least~$3$ vertices. 
\end{fact}

Note that Fact~\ref{fact:Ajk} holds accordingly for $A_{jk}'$. 
From the auxiliary graphs $A_{jk}$ and $A'_{jk}$, we construct other auxiliary graphs $H_j$ and $H_j'$, $1\le j\le m$, as follows. 
$H_j$ is obtained from $A_{j1}, A_{j2}$ and $A_{j3}$ by adding three additional edges $d_{j1}y_{j2}$, $d_{j2}y_{j3}$ and $d_{j3}y_{j1}$. $H_j'$ is similarly defined; see also Fig.~\ref{fig:Hj}. 

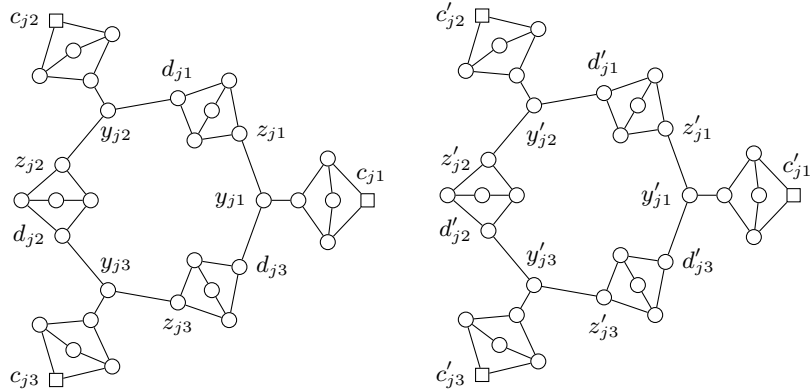
\begin{figure}[!ht]
\begin{center}
\tikzstyle{vertexS}=[draw,circle,inner sep=2pt,fill=black]
\tikzstyle{squareS}=[draw,rectangle,inner sep=2.5pt,fill=black]
\tikzstyle{subdS}=[draw,circle,inner sep=1.5pt,fill=black]
\tikzstyle{square}=[draw, rectangle,inner sep=2.5pt]
\tikzstyle{vertex}=[draw,circle,inner sep=2pt] 
\tikzstyle{subd}=[draw,circle,inner sep=1.5pt]
\begin{tikzpicture}[scale=.46] 
%\draw [color=gray!50] (0,0) grid (15,15);

\node[square] (cj1) at ({3+6*cos(0)},{4+6*sin(0)}) [label=above:\small $~c_{j1}$] {};
\node[vertex] (pj12) at ({3+5*cos(0)},{4+5*sin(0)}) {}; %[label=below:\small $pj12$] {};
\node[vertex] (xj1) at ({3+4*cos(0)},{4+4*sin(0)}) {}; %[label=above:\small $xj1$] {};
\node[vertex] (yj1) at ({3+3*cos(0)},{4+3*sin(0)}) [label=left:\small $y_{j1}$] {};
\node[vertex] (pj11) at ({3+5*cos(14)},{4+5*sin(14)}) {}; %[label=below:\small $pj11$] {};
\node[vertex] (pj13) at ({3+5*cos(346)},{4+5*sin(346)}) {}; %[label=below:\small $pj13$] {};
\node[vertex] (zj1) at ({3+3*cos(40)},{4+3*sin(40)}) [label=right:\small $z_{j1}$] {};
\node[vertex] (dj1) at ({3+3*cos(80)},{4+3*sin(80)}) [label=above:\small $d_{j1}$] {};
\node[vertex] (qj11) at ({3+4*cos(60)},{4+4*sin(60)}) {}; %[label=right:\small $q_{j11}$] {};
\node[vertex] (qj12) at ({3+3*cos(60)},{4+3*sin(60)}) {}; %[label=right:\small $q_{j12}$] {};
\node[vertex] (qj13) at ({3+2*cos(60)},{4+2*sin(60)}) {}; %[label=right:\small $q_{j13}$] {};

\node[square] (cj2) at ({3+6*cos(120)},{4+6*sin(120)}) [label=left:\small $c_{j2}$] {};
\node[vertex] (pj22) at ({3+5*cos(120)},{4+5*sin(120)}) {}; %[label=below:\small $pj22$] {};
\node[vertex] (xj2) at ({3+4*cos(120)},{4+4*sin(120)}) {}; %[label=above:\small $xj2$] {};
\node[vertex] (yj2) at ({3+3*cos(120)},{4+3*sin(120)}) [label=below:\small $~~y_{j2}$] {};
\node[vertex] (pj21) at ({3+5*cos(134)},{4+5*sin(134)}) {}; %[label=below:\small $pj21$] {};
\node[vertex] (pj23) at ({3+5*cos(106)},{4+5*sin(106)}) {}; %[label=below:\small $pj23$] {};
\node[vertex] (zj2) at ({3+3*cos(160)},{4+3*sin(160)}) [label=left:\small $z_{j2}$] {};
\node[vertex] (dj2) at ({3+3*cos(200)},{4+3*sin(200)}) [label=left:\small $d_{j2}$] {};
\node[vertex] (qj21) at ({3+4*cos(180)},{4+4*sin(180)}) {}; %[label=right:\small $q_{j21}$] {};
\node[vertex] (qj22) at ({3+3*cos(180)},{4+3*sin(180)}) {}; %[label=right:\small $q_{j22}$] {};
\node[vertex] (qj23) at ({3+2*cos(180)},{4+2*sin(180)}) {}; %[label=right:\small $q_{j23}$] {};

\node[square] (cj3) at ({3+6*cos(240)},{4+6*sin(240)}) [label=left:\small $c_{j3}$] {};
\node[vertex] (pj32) at ({3+5*cos(240)},{4+5*sin(240)}) {}; %[label=below:\small $pj32$] {};
\node[vertex] (xj3) at ({3+4*cos(240)},{4+4*sin(240)}) {}; %[label=above:\small $xj3$] {};
\node[vertex] (yj3) at ({3+3*cos(240)},{4+3*sin(240)}) [label=above:\small $~~y_{j3}$] {};
\node[vertex] (pj31) at ({3+5*cos(226)},{4+5*sin(226)}) {}; %[label=below:\small $pj31$] {};
\node[vertex] (pj33) at ({3+5*cos(254)},{4+5*sin(254)}) {}; %[label=below:\small $pj33$] {};
\node[vertex] (zj3) at ({3+3*cos(280)},{4+3*sin(280)}) [label=below:\small $z_{j3}$] {};
\node[vertex] (dj3) at ({3+3*cos(320)},{4+3*sin(320)}) [label=right:\small $d_{j3}$] {};
\node[vertex] (qj31) at ({3+4*cos(300)},{4+4*sin(300)}) {}; %[label=right:\small $q_{j31}$] {};
\node[vertex] (qj32) at ({3+3*cos(300)},{4+3*sin(300)}) {}; %[label=right:\small $q_{j32}$] {};
\node[vertex] (qj33) at ({3+2*cos(300)},{4+2*sin(300)}) {}; %[label=right:\small $q_{j33}$] {};

\draw (cj1)--(pj11)--(xj1)--(yj1)--(zj1)--(qj11)-- (dj1)--(qj13)--(zj1);
\draw (cj1)--(pj13)--(xj1);
\draw (pj11)--(pj12)--(pj13);
\draw (qj11)--(qj12)--(qj13);

\draw (cj2)--(pj21)--(xj2)--(yj2)--(zj2)--(qj21)-- (dj2)--(qj23)--(zj2);
\draw (cj2)--(pj23)--(xj2);
\draw (pj21)--(pj22)--(pj23);
\draw (qj21)--(qj22)--(qj23);

\draw (cj3)--(pj31)--(xj3)--(yj3)--(zj3)--(qj31)-- (dj3)--(qj33)--(zj3);
\draw (cj3)--(pj33)--(xj3);
\draw (pj31)--(pj32)--(pj33);
\draw (qj31)--(qj32)--(qj33);

\draw (dj1)--(yj2); \draw (dj2)--(yj3); \draw (dj3)--(yj1);
\end{tikzpicture} 
\quad
\begin{tikzpicture}[scale=.46] 
%\draw [color=gray!50] (0,0) grid (15,15);

\node[square] (cj1) at ({3+6*cos(0)},{4+6*sin(0)}) [label=above:\small $~c_{j1}'$] {};
\node[vertex] (pj12) at ({3+5*cos(0)},{4+5*sin(0)}) {}; %[label=below:\small $pj12$] {};
\node[vertex] (xj1) at ({3+4*cos(0)},{4+4*sin(0)}) {}; %[label=above:\small $xj1$] {};
\node[vertex] (yj1) at ({3+3*cos(0)},{4+3*sin(0)}) [label=left:\small $y_{j1}'$] {};
\node[vertex] (pj11) at ({3+5*cos(14)},{4+5*sin(14)}) {}; %[label=below:\small $pj11$] {};
\node[vertex] (pj13) at ({3+5*cos(346)},{4+5*sin(346)}) {}; %[label=below:\small $pj13$] {};
\node[vertex] (zj1) at ({3+3*cos(40)},{4+3*sin(40)}) [label=right:\small $z_{j1}'$] {};
\node[vertex] (dj1) at ({3+3*cos(80)},{4+3*sin(80)}) [label=above:\small $d_{j1}'$] {};
\node[vertex] (qj11) at ({3+4*cos(60)},{4+4*sin(60)}) {}; %[label=right:\small $q_{j11}$] {};
\node[vertex] (qj12) at ({3+3*cos(60)},{4+3*sin(60)}) {}; %[label=right:\small $q_{j12}$] {};
\node[vertex] (qj13) at ({3+2*cos(60)},{4+2*sin(60)}) {}; %[label=right:\small $q_{j13}$] {};

\node[square] (cj2) at ({3+6*cos(120)},{4+6*sin(120)}) [label=left:\small $c_{j2}'$] {};
\node[vertex] (pj22) at ({3+5*cos(120)},{4+5*sin(120)}) {}; %[label=below:\small $pj22$] {};
\node[vertex] (xj2) at ({3+4*cos(120)},{4+4*sin(120)}) {}; %[label=above:\small $xj2$] {};
\node[vertex] (yj2) at ({3+3*cos(120)},{4+3*sin(120)}) [label=below:\small $~~y_{j2}'$] {};
\node[vertex] (pj21) at ({3+5*cos(134)},{4+5*sin(134)}) {}; %[label=below:\small $pj21$] {};
\node[vertex] (pj23) at ({3+5*cos(106)},{4+5*sin(106)}) {}; %[label=below:\small $pj23$] {};
\node[vertex] (zj2) at ({3+3*cos(160)},{4+3*sin(160)}) [label=left:\small $z_{j2}'$] {};
\node[vertex] (dj2) at ({3+3*cos(200)},{4+3*sin(200)}) [label=left:\small $d_{j2}'$] {};
\node[vertex] (qj21) at ({3+4*cos(180)},{4+4*sin(180)}) {}; %[label=right:\small $q_{j21}$] {};
\node[vertex] (qj22) at ({3+3*cos(180)},{4+3*sin(180)}) {}; %[label=right:\small $q_{j22}$] {};
\node[vertex] (qj23) at ({3+2*cos(180)},{4+2*sin(180)}) {}; %[label=right:\small $q_{j23}$] {};

\node[square] (cj3) at ({3+6*cos(240)},{4+6*sin(240)}) [label=left:\small $c_{j3}'$] {};
\node[vertex] (pj32) at ({3+5*cos(240)},{4+5*sin(240)}) {}; %[label=below:\small $pj32$] {};
\node[vertex] (xj3) at ({3+4*cos(240)},{4+4*sin(240)}) {}; %[label=above:\small $xj3$] {};
\node[vertex] (yj3) at ({3+3*cos(240)},{4+3*sin(240)}) [label=above:\small $~~y_{j3}'$] {};
\node[vertex] (pj31) at ({3+5*cos(226)},{4+5*sin(226)}) {}; %[label=below:\small $pj31$] {};
\node[vertex] (pj33) at ({3+5*cos(254)},{4+5*sin(254)}) {}; %[label=below:\small $pj33$] {};
\node[vertex] (zj3) at ({3+3*cos(280)},{4+3*sin(280)}) [label=below:\small $z_{j3}'$] {};
\node[vertex] (dj3) at ({3+3*cos(320)},{4+3*sin(320)}) [label=right:\small $d_{j3}'$] {};
\node[vertex] (qj31) at ({3+4*cos(300)},{4+4*sin(300)}) {}; %[label=right:\small $q_{j31}$] {};
\node[vertex] (qj32) at ({3+3*cos(300)},{4+3*sin(300)}) {}; %[label=right:\small $q_{j32}$] {};
\node[vertex] (qj33) at ({3+2*cos(300)},{4+2*sin(300)}) {}; %[label=right:\small $q_{j33}$] {};

\draw (cj1)--(pj11)--(xj1)--(yj1)--(zj1)--(qj11)-- (dj1)--(qj13)--(zj1);
\draw (cj1)--(pj13)--(xj1);
\draw (pj11)--(pj12)--(pj13);
\draw (qj11)--(qj12)--(qj13);

\draw (cj2)--(pj21)--(xj2)--(yj2)--(zj2)--(qj21)-- (dj2)--(qj23)--(zj2);
\draw (cj2)--(pj23)--(xj2);
\draw (pj21)--(pj22)--(pj23);
\draw (qj21)--(qj22)--(qj23);

\draw (cj3)--(pj31)--(xj3)--(yj3)--(zj3)--(qj31)-- (dj3)--(qj33)--(zj3);
\draw (cj3)--(pj33)--(xj3);
\draw (pj31)--(pj32)--(pj33);
\draw (qj31)--(qj32)--(qj33);

\draw (dj1)--(yj2); \draw (dj2)--(yj3); \draw (dj3)--(yj1);
\end{tikzpicture} 
\caption{The auxiliary graph $H_j$ (left) and $H_j'$ (right).}\label{fig:Hj}
\end{center}
\end{figure} 

\begin{fact}\label{fact:Hj}
\begin{itemize}
\item[\em (a)] 
Any claw deletion set in $H_j$ has at least $8$ vertices. 

\item[\em (b)] For any non-empty proper subset $T\subset\{c_{j1}, c_{j2}, c_{j3}\}$ there is an optimal claw deletion set of size~$8$ in $H_j$ that contains $T$. 

\item[\em (c)] No optimal claw deletion set in $H_j$ contains all $c_{j1}, c_{j2}$ and $c_{j3}$, as well as a neighbor of any $c_{jk}$. 
%Moreover, if $c_{jk}'\not\in S$ then the two neighbors of $c_{jk}'$ are outside $S$.
\end{itemize}
\end{fact}
\begin{proof}
(a): Observe that, for each $1\le k\le 3$, $d_{jk}$ and its neighbors induce a claw which does not contain an $x$ nor a $z$-vertex. 
Since a vertex of this claw must belong to any claw deletion set $S$, we have with Fact~\ref{fact:Ajk}, that two of the restrictions $S_{j1}, S_{j2}$ and $S_{j3}$ of $S$ on $A_{j1}, A_{j2}$ and $A_{j3}$, respectively, have size at least $3$. Hence $|S|\ge |S_{j1}|+|S_{j2}|+|S_{j3}|\ge 8$.

(b): %By inspection.
$\{c_{j1}$, $y_{j1}$, $d_{j1}$, $x_{j2}$, $z_{j2}$, $x_{j3}$, $y_{j3}$, $d_{j3}\}$ is an optimal claw deletion set containing $c_{j1}$; similar for $c_{j2}$ and for $c_{j3}$.  
$\{x_{j1}$, $z_{j1}$, $c_{j2}$, $y_{j2}$, $d_{j2}$, $c_{j3}$, $y_{j3}$, $d_{j3}\}$ is an optimal claw deletion set containing $c_{j2}$ and $c_{j3}$; similar for the pairs $c_{j1}$, $c_{j2}$ and $c_{j1}$, $c_{j3}$.   

(c): If all $c_{j1}, c_{j2}$ and $c_{j3}$ are in an optimal claw deletion set $S$ then, by Fact~\ref{fact:Ajk}, $S$ has at least $3$ vertices in each $A_{jk}$, and hence $|S|\ge 9$, contradicting (b). For the last part, let $S_{j1}, S_{j2}$ and $S_{j3}$ be the restrictions of $S$ on $A_{j1}, A_{j2}$ and $A_{j3}$, respectively. 
By (a) and (b), one of these sets is of size $2$ and the other have size $3$. Let $|S_{j1}|=2$ and $|S_{j2}|= 3, |S_{j3}|= 3$, say. 
By Fact~\ref{fact:Ajk}, $S_{j1}=\{x_{j1}, z_{j2}\}$. Then $y_{j2}$ must belong to $S_{j2}$ and it can be checked that $S_{j2}$ cannot contain any neighbor of $c_{j2}$. 
Suppose $S_{j3}$ contains a neighbor of $c_{j3}$. It follows that $d_{j3}$ and its two neighbors in $A_{j3}$ are not in $S_{j3}$, and thus $d_{j3}$ is the center of a claw (containing $y_{j1}$) outside $S$, a contradiction.   
\qed
\end{proof}

Note that Fact~\ref{fact:Hj} holds accordingly for $H_j'$. 
We are now ready to describe the clause gadget. 

\textbf{Clause gadget.} For each clause $C_j$ consisting of variables $c_{j1}, c_{j2}$ and $c_{j3}$, let 
 $G(C_j)$ be the graph consisting of two connected components, $H_j$ and $H_j'$. We call the six vertices $c_{jk}$ and $c_{jk}'$, $1\le k\le 3$, the \emph{clause vertices}. 

From Fact~\ref{fact:Hj}, we immediately have:

\begin{fact}\label{fact:claw-cj}
$G(C_j)$ admits an optimal claw deletion set of $16$ vertices. 
No optimal claw deletion set in $G(C_j)$ contains all $c_{j1}, c_{j2}$ and $c_{j3}$, or all $c_{j1}', c_{j2}'$ and $c_{j3}'$, or a neighbor of any $c_{jk}$ or a neighbor of any $c_{jk}'$. 
%Moreover, if $c_{jk}'\not\in S$ then the two neighbors of $c_{jk}'$ are outside $S$.
\end{fact}

Finally, the graph $G$ is obtained by connecting the variable and clause gadgets as follows: if variable $v_i$ appears in clause $C_j$, i.e., $v_i=c_{jk}$ for some $k\in\{1,2,3\}$, then 

\begin{itemize}
\item connect the variable vertex $v_{ij}$ in $G(v_i)$ and the clause vertex $c_{jk}$ in $G(C_j)$ by an edge; $v_{ij}$ is the \emph{corresponding variable vertex} of the clause vertex $c_{jk}$, and
\item connect the variable vertex $v_{ij}'$ in $G(v_i)$ and the clause vertex $c_{jk}'$ in $G(C_j)$ by an edge; $v_{ij}'$ is the \emph{corresponding variable vertex} of the clause vertex $c_{jk}'$.
\end{itemize} 

%It follows from construction, that $G$ has maximum degree $3$. 

\medskip
\begin{fact}\label{fact:claw-bip}
$G$ has maximum degree~$3$ and is bipartite.
\end{fact}
\begin{proof}
It follows from the construction that $G$ has maximum degree~$3$. 
To see that $G$ is bipartite, note first that the bipartite graph forming by all variable gadgets $G(v_i)$ has a bipartition $(A,B)$ such that all $v_{ij}$ are in $A$ and all $v_{ij}'$ are in $B$, and the bipartite graph forming by all clause gadgets $G(C_j)$ has a bipartition $(C,D)$ such that all $c_{j1}, c_{j2}, c_{j3}$ are in $C$ and all $c_{j1}', c_{j2}', c_{j3}'$ are in~$D$. Hence, by construction, $(A\cup D, B\cup C)$ is a bipartition of $G$.\qed
\end{proof}

\begin{fact}\label{fact:claw-cj1}
Let $S$ be a claw deletion set of $G$ such that $S$ contains exactly $2m$ vertices from each $G(v_i)$. 
Suppose that, for some $1\le j\le m$, $S$ contains none of $c_{j1}, c_{j2}$ and $c_{j3}$, or none of $c_{j1}', c_{j2}'$ and $c_{j3}'$. Then $S$ contains at least~$17$ vertices from $G(C_j)$. 
\end{fact}
\begin{proof}
Consider the case that $S$ contains none of $c_{j1}, c_{j2}$ and $c_{j3}$; the other case is similar. Suppose for the contrary that $S$ contains at most $16$ vertices from $G(C_j)$. Then by Facts~\ref{fact:claw-cj} and~\ref{fact:Hj}, the restrictions of $S$ on $H_j$ and on $H_j'$ are optimal claw deletion sets in $H_j$ and $H_j'$, respectively. 
 
Let $v_{rj}$, $v_{sj}$ and~$v_{tj}$ be the corresponding variable vertices of $c_{j1}$, $c_{j2}$ and~$c_{j3}$, respectively. 
By Fact~\ref{fact:Hj}~(c), all neighbors of $c_{j1}$, $c_{j2}$ and~$c_{j3}$ in $H_j$ are outside $S$, implying all $v_{rj}$, $v_{sj}$ and~$v_{tj}$ are in $S$. 
By Fact~\ref{fact:claw-vi2}~(a), all $v_{rj}'$, $v_{sj}'$ and~$v_{tj}'$ together with their neighbors in $G(v_r)$, $G(v_s)$ and $G(v_t)$ are outside $S$, implying all $c_{j1}', c_{j2}'$ and $c_{j3}'$ are in $S$. 
This is a contradiction to Fact~\ref{fact:Hj}~(c). 
\qed
\end{proof}

%Set $k=2nm+16m$. %We now show that $F\in \text{\PNAE3SAT}$ if and only if $(G,k)\in\text{\CVD}$.  
Now suppose that $G$ has a claw deletion set $S$ with at most $2nm+16m$ vertices. 
Then by Facts~\ref{fact:claw-vi1} and~\ref{fact:Hj}~(a), $S$ has exactly $2nm+16m$ vertices, and~$S$ contains exactly $2m$ vertices from each $G(v_i)$ and exactly $16$ vertices from each $G(C_j)$. 
%Moreover, by Fact~\ref{obs:cj} and Fact~\ref{fact:cj1}~(a), we have: For any $1\le j\le m$, if both $c_{j1}$ and $c_{j1}'$ are in $S$ then exactly one of $c_{j2}, c_{j2}'$ is in $S$ and exactly one of $c_{j3}, c_{j3}'$ is in $S$. Analogously in case both $c_{j2}$ and $c_{j2}'$ are in $S$, respectively, both $c_{j2}$ and $c_{j2}'$ are in $S$. 

Consider the truth assignment in which a variable $v_i$ is true if all its associated variable vertices $v_{ij}$, $1\le j\le m$, are in $S$. Note that by Fact~\ref{fact:claw-vi2}~(a) and~(b), this assignment is well-defined. For each $G(C_j)$, it follows from Facts~\ref{fact:claw-cj} and~\ref{fact:claw-cj1} that some of $c_{j1}, c_{j2}, c_{j3}$ is outside $S$ and some of $c_{j1}', c_{j2}', c_{j3}'$ is outside $S$ as well. 
That is, the clause $C_j$ contains a true variables and a false variable.  %arguments similar to cluster case.
Thus, if $G$ admits a claw deletion set $S$ with at most $2nm+16m$ then $F$ has a nae assignment.  

Conversely, suppose that there is a nae assignment for~$F$. Then a claw deletion set $S$ of size $2nm+16m$ for $G$ is as follows. 
For each variable $v_i$, $1\le i\le n$: 
\begin{itemize}
\item If variable $v_i$ is true, then put all $2m$ vertices $v_{ij}, b_{ij}^2$, $1\le j\le m$, into~$S$. 
\item If variable $v_i$ is false, then put all $2m$ vertices $v_{ij}', a_{ij}^2$, $1\le j\le m$, into~$S$.  
\end{itemize}

For each clause $C_j$, $1\le j\le m$, let $v_{rj}$, $v_{sj}$ and $v_{tj}$ be the corresponding variable vertices of $c_{j1}$, $c_{j2}$ and $c_{j3}$, respectively. Extend $S$ to~$16$ vertices of~$G(C_j)$ as follows; cf. Fact~\ref{fact:Hj}~(b). 

\begin{itemize}
\item If $C_j$ has one true variable and two false variables, say $v_r$ is true, $v_s$ and $v_t$ are false, then put 
$x_{j1}$, $z_{j1}$, $c_{j2}$, $y_{j2}$, $d_{j2}$, $c_{j3}$, $y_{j3}$, $d_{j3}$, 
$c_{j1}'$, $y_{j1}'$, $d_{j1}'$, $x_{j2}'$, $z_{j2}'$, $x_{j3}'$, $y_{j3}'$ and $d_{j3}'$ into $S$. 
% and extend to other~$8$ vertices as indicated in Fig.~\ref{fig:claw-extendGCj} on the left hand.
\item If $C_j$ has two true variables and one false variable, say $v_r$ and $v_s$ are true, $v_t$ is false, then put 
$x_{j1}$, $z_{j1}$, $x_{j2}$, $y_{j2}$, $d_{j2}$, $c_{j3}$, $y_{j3}$, $d_{j3}$, 
$c_{j1}'$, $y_{j1}'$, $d_{j1}'$, $c_{j2}'$, $y_{j2}'$, $d_{j2}'$, $x_{j3}'$ and $z_{j3}'$ into $S$. 
% and extend to other~$8$ vertices as indicated in Fig.~\ref{fig:claw-extendGCj} on the right hand.
\end{itemize}
  
%See also Fig.~\ref{fig:claw-extendGCj}.  

By construction, $S$ has exactly $n\times 2m+m\times 16$ vertices and, for every $i$ and $j$, $G(v_i)-S$ and $G(C_j)-S$ are claw-free. Therefore, since a clause vertex is in $S$ if and only if the corresponding variable vertex is not in $S$, $G-S$ is claw-free. Thus, if $F$ can be satisfied by a nae assignment then $G$ has a claw deletion set of size at most $2nm+16m$.

In summary, we obtain:
\begin{theorem}\label{thm:dclaw-bip+deg-3}
For any $d\ge 3$, \dclawVD{d} is $\NP$-complete even when restricted to bipartite %planar?
graphs of maximum degree~$d$.
\end{theorem}
\begin{proof}
The case $d=3$ follows from the previous arguments. For an fixed integer $d>3$ and a bipartite graph $G$ of maximum degree~$3$, let $G'$ be obtained from $G$ by adding, for each vertex~$v$ in $G$, $d-3$ new vertex $v_1,\ldots, v_{d-3}$ all are adjacent to exactly~$v$. 
Then~$G'$ is bipartite and has maximum degree~$d$. It can be verified immediately that $G$ has a claw deletion set of size at most~$k$ if and only if $G'$ has a $d$-claw deletion set of size at most~$k$. \qed   
\end{proof}

\subsubsection{Bipartite graphs of bounded diameter} 
%---------------------------------------------------

From Theorems~\ref{thm:bip+deg-3} and \ref{thm:dclaw-bip+deg-3}, and (the proof of) Observation~\ref{obs:diam3} we conclude:
\begin{theorem}\label{thm:bip+diam3}
For any $d\ge 2$, \dclawVD{d} is $\NP$-complete even when restricted to bipartite graphs of diameter~$3$ with only two unbounded vertices.
\end{theorem}

We remark that \mindclawVD{d} is polynomially solvable on bipartite graphs of diameter at most two. This can be seen as follows.  
Let $G=(X,Y,E)$ be a bipartite of of diameter $\le 2$; such a bipartite graph is complete bipartite. 
Note first that $X$ and $Y$ are $d$-claw deletion sets for $G$. We will see that any optimal $d$-claw deletion set is $X$ or $Y$ or is of the form $(X\setminus X')\cup (Y\setminus Y')$ for some $d-1$-element sets $X'\subseteq X$ and $Y'\subseteq Y$. (In particular, all optimal $d$-claw deletion sets can be found in $O(n^{d-1})$ time.) 
Indeed, let $S$ be an optimal $d$-claw deletion set.  
If $X\subseteq S$, then by the optimality of $S$, $S=X$. Similarly, if $Y\subseteq S$, then $S=Y$. 
So, let $X'=X\setminus S\not=\emptyset$ and $Y'=Y\setminus S\not=\emptyset$. Then 
$|X'|\le d-1$ and $|Y'|\le d-1$: if $|X'|\ge d$, say, then any vertex in $Y'$ and $d$ vertices in $X'$ together would induce a $d$-claw in $G-S$. Thus, by the optimality of $S$, $|X'|=|Y'|=d-1$. 

Unfortunately, we have to leave open the complexity of \dclawVD{d} on bipartite graphs of diameter~$3$ with only one vertex of unbounded degree.

%%Bipartite graphs of diameter 3 with only one unbounded vertex ???!

\subsection{Split graphs}\label{subsec:split}
%============================================
In this subsection, we show that, for any $d\ge 3$, \dclawVD{d} is $\NP$-complete even when restricted to split graphs. Note that split graphs have diameter~$3$. By Observation~\ref{obs:diam2}, however, we will see that \dclawVD{d} is hard even on split graphs of diameter~$2$.
Recall that \dclawVD{1} and \dclawVD{2} are solvable in polynomial time on split graphs.

Let $d\ge 3$ be a fixed integer. We reduce \dHVC{(d-1)} to \dclawVD{d}. Our reduction is inspired by the reduction from \VC\ to \dclawVD{3} in~\cite{Bonomo-Braberman20}.
Let $G=(V,E)$ be a $(d-1)$-uniform hypergraph with $n=|V|$ vertices and $m=|E|$ edges. We may assume that for any hyperedge $e\in E$ there is another hyperedge $f\in E$ such that $e\cap f=\emptyset$. For otherwise, $G$ has a vertex cover of size $\le |e|=d-1$ and therefore \dHVC{(d-1)} is polynomially solvable on such inputs $G$.
We construct a split graph $G'=(V',E')$ with $V'=Q\cup I$, where $Q$ is a clique and $I$ is an independent set, as follows:
\begin{itemize}
\item $I=\{v'\mid v\in V\}$;
\item for each edge $e\in E$, let $Q(e)$ be a clique of size $n$;
%\item let $D$ be a dummy clique of size $n$;
\item all sets $I$ and $Q(e)$, $e\in E$, are pairwise disjoint;
\item make $\bigcup_{e\in E} Q(e)$ to  clique $Q$;
\item for each $v'\in I$ and $e\in E$, connect $v'$ to all vertices in $Q(e)$ if and only if $v\in e$.
\end{itemize}
The description of the split graph $G'$ is complete. Note that $G'$ has $nm+n$ vertices and $O(n^2m^2)$ edges, and can be constructed in $O(n^2m^2)$ time.

For each $e\in E$, write $e'=\{v'\in I\mid v\in e\}$. By construction, every vertex in $Q(e)$ has exactly $d-1$ neighbors in $I$, namely the vertices in $e'$. Hence, every induced $d$-claw in $G'$ is formed by a center vertex $x\in Q(e)$ for some $e\in E$ and $e'\cup\{y\}$, where $y$ is any vertex in $Q(f)$, $f\in E$, such that $f\cap e=\emptyset$.
It follows that $G'$ contains no induced $(d+1)$-claws.

\smallskip
\begin{fact}\label{Claim1}
If $S$ is a vertex cover in the hypergraph $G$, then $S'=\{v'\mid v\in S\}$ is a $d$-claw deletion set in the split graph~$G'$.
\end{fact}
\begin{proof}
If $C$ is a $d$-claw in $G'$ with center vertex $x\in Q(e)$ for some $e\in E$ such that $C\cap S'=\emptyset$, then $e'\cap S'=\emptyset$. This means $e\cap S=\emptyset$, contradicting the fact that $S$ is a vertex cover of the hypergraph $G$.\qed
\end{proof}

\begin{fact}\label{Claim2}
If $S'$ is a $d$-claw deletion set in the split graph $G'$ of size $<n$, then $S=\{v\mid v'\in S'\}$ is a vertex cover in the hypergraph $G$.
\end{fact}
\begin{proof}
First, for each $e\in E$, $S'\cap e'\not=\emptyset$. For otherwise let $S'\cap e'=\emptyset$ for some $e\in E$. Since $|S'|< n$, there is a vertex $x\in Q(e)\setminus S'$ and a vertex $y\in Q(f)\setminus S'$ with $f\cap e=\emptyset$. Then $x, y$ and $e'$ induce a $d$-claw in $G'-S'$, a contradiction.
We have seen that, for each $e\in E$, $S'\cap e'\not=\emptyset$. Then, with $S=\{v\mid v'\in S'\}$, we have $S\cap e\not=\emptyset$ for all $e\in E$. That is, $S$ is a vertex cover of the hypergraph $G$. % with $|S|\le|S'|\le k$.
\qed
\end{proof}

\begin{fact}\label{Claim3}
The size of a smallest vertex cover of $G$, $\textsc{opt}_{\textsc{vc}}(G)$, and the size of a smallest $d$-claw deletion set in $G'$, $\textsc{opt}_{\dclawVD{d}}(G')$, are equal.
\end{fact}
\begin{proof} By Fact~\ref{Claim1}, $\textsc{opt}_{\dclawVD{d}}(G')\le \textsc{opt}_{\textsc{vc}}(G)$.
Let $S'$ be a smallest $d$-claw deletion set in $G'$. Then $|S'|< n$ because $I$ minus an arbitrary vertex is a $d$-claw deletion set in~$G'$ with $n-1$ vertices.
Hence, by Fact~\ref{Claim2}, $S=\{v\mid v'\in S'\}$ is a vertex cover in the hypergraph $G$ with $|S|\le |S'|$. Thus, $\textsc{opt}_{\textsc{vc}}(G)\le \textsc{opt}_{\dclawVD{d}}(G')$.\qed
\end{proof}

We now derive hardness results for \dclawVD{d} and \mindclawVD{d} from the above reduction.
\begin{theorem}\label{thm:1}
For any fixed integer $d\ge 3$, \dclawVD{d} is $\NP$-complete, even when restricted to
\begin{itemize}
\item[\em (i)] split graphs without induced $(d+1)$-claws, and
\item[\em (ii)] split graphs of diameter~$2$.
\end{itemize}
\end{theorem}
\begin{proof}
Part (i) follows from Facts~\ref{Claim1} and~\ref{Claim2}, and the fact that the split graph $G'$ contains no induced $(d+1)$-claws.
Part~(ii) follows from~(i) and (the proof of) Observation~\ref{obs:diam2}.
\qed
\end{proof}

We remark that both hardness results in Theorem~\ref{thm:1} are optimal in the sense that \dclawVD{d} is trivial for graphs without induced $d$-claws, in particular for graphs of diameter~$1$, \emph{i.e.}, complete graphs.
We also remark that Theorem~\ref{thm:1} implies, in particular, that \dclawVD{d} is $\NP$-complete on chordal graphs for any $d\ge3$, while the complexity  of \dclawVD{2} on chordal graphs is still open (cf.~\cite{CaoKOY17,0001KOY18}).

Since it is UGC-hard to approximate \mindHVC{(d-1)} to a factor $(d-1)-\epsilon$ for any $\epsilon>0$~\cite{KhotR08},  %(see~\cite{BansalK10,KhotR08}),
Fact~\ref{Claim3} implies:

\begin{theorem}\label{thm:2}
Let $d\ge 3$ be a fixed integer. Assuming the UGC, there is no approximation algorithm for \mindclawVD{d} within a factor better than $d-1$, even when restricted to split graphs without induced $(d+1)$-claws.
\end{theorem}

We remark that for triangle-free graphs, in particular bipartite graphs, \mindclawVD{d} and \textsc{$d$-claw-transversal} coincide, hence a result in~\cite{GuruswamiL17} implies that \mindclawVD{d} admits an $O(\log(d+1))$-approximation when restricted to bipartite graphs.

\section{A polynomially solvable case}\label{sec:d-block}
%=======================================================
In this section, we will show a polynomial-time algorithm solving \mindclawVD{d} for what we call $d$-block graphs.
As $d$-block graphs generalize block graphs, this result extends the polynomial-time algorithm for \dclawVD{2} on block graphs given in~\cite{0001KOY18} to \dclawVD{d} for all $d\ge 2$ on block graphs, and improves the polynomial-time algorithm for \mindclawVD{3} given in~\cite{Bonomo-Braberman20} on block graphs to $3$-block graphs.

Recall that a \emph{block} in a graph is a maximal biconnected subgraph. \emph{Block graphs} are those in which every block is a clique.
For each integer $d\ge2$, the $d$-block graphs defined below generalize block graphs.

\begin{definition}
Let $d\ge 2$ be an integer. A graph $G$ is \emph{$d$-block graph} if, for every block $B$ of $G$,
\begin{itemize}
\item $B$ is $d$-claw free,
\item for every cut vertex $v$ of $G$, $N(v)\cap B$ is a clique, and
\item the cut vertices of $G$ in $B$ induce a clique.
\end{itemize}
\end{definition}

Note that block graphs are exactly the $2$-block graphs and $d$-block graphs are $(d+1)$-block graphs, but not the converse.
Note also that, for $d\ge 3$, $d$-block graphs need not to be chordal since they may contain arbitrary long induced cycles.
An example of a $3$-block graph is shown in~Fig.~\ref{fig:d-block}.

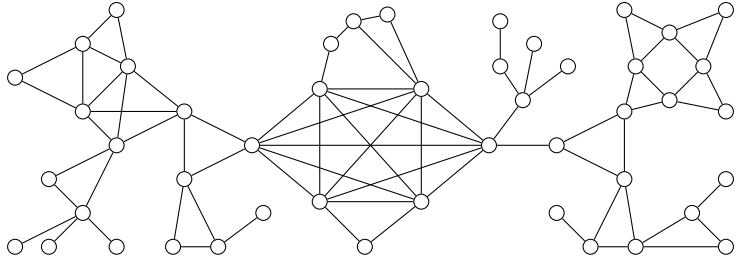
\begin{figure}[ht]
% \tikzstyle{vertex}=[draw,circle,inner sep=1.5pt]
% \tikzstyle{vertexX}=[circle,inner sep=1.5pt,fill=black]
% \tikzstyle{vertexY}=[draw,circle,inner sep=1.5pt,fill=lightgray]

\tikzstyle{vertexS}=[draw,circle,inner sep=2pt,fill=black]
\tikzstyle{squareS}=[draw,rectangle,inner sep=2.5pt,fill=black]
\tikzstyle{subdS}=[draw,circle,inner sep=1.5pt,fill=black]
\tikzstyle{square}=[draw, rectangle,inner sep=2.5pt]
\tikzstyle{vertex}=[draw,circle,inner sep=2pt] 
\tikzstyle{subd}=[draw,circle,inner sep=1.5pt]
\begin{center}
\begin{tikzpicture}[scale=.3] %25]
\node[vertex] (1) at (5.5,5.5) {};
\node[vertex] (2) at (4,2.5) {};
\node[vertex] (3) at (8.5,7) {};
\node[vertex] (4) at (8.5,4) {};
\node[vertex] (5) at (10,1) {};
\node[vertex] (6) at (11.5,5.5) {};
\node[vertex] (7) at (22,5.5) {};
\node[vertex] (8) at (25,5.5) {};
\node[vertex] (9) at (28,7) {};
\node[vertex] (10) at (28,4) {};
\node[vertex] (11) at (26.5,1) {};
\node[vertex] (12) at (28.5,1) {};
\node[vertex] (13) at (31,2.5) {};

\node[vertex] (v1) at (1,1) {};
\node[vertex] (v2) at (5.5,1) {};
\node[vertex] (v3) at (2.5,4) {};
\node[vertex] (v4) at (4,7) {};
\node[vertex] (v5) at (1,8.5) {};
\node[vertex] (v6) at (4,10) {};
\node[vertex] (v7) at (5.5,11.5) {};
\node[vertex] (v8) at (6,9) {};
\node[vertex] (v9) at (8,1) {};
\node[vertex] (u1) at (12,2.5) {};
\node[vertex] (u2) at (16.5,1) {};
\node[vertex] (u3) at (14.5,3) {};
\node[vertex] (u4) at (14.5,8) {};
\node[vertex] (u5) at (15,10) {};
\node[vertex] (u6) at (16,11) {};
\node[vertex] (u7) at (17.5,11.3) {};
\node[vertex] (u8) at (19,8) {};
\node[vertex] (u9) at (19,3) {};
\node[vertex] (x1) at (25,2.5) {};
\node[vertex] (x2) at (23.5,7.5) {};
\node[vertex] (x3) at (25.5,9) {};
\node[vertex] (x4) at (28.5,9) {};
\node[vertex] (x5) at (28,11.5) {};
\node[vertex] (x6) at (30,10.5) {};
\node[vertex] (x7) at (32.5,11.5) {};
\node[vertex] (x8) at (31.5,9) {};
\node[vertex] (x9) at (32.5,7) {};
\node[vertex] (y1) at (30,7.5) {};
\node[vertex] (y2) at (32.5,4) {};
\node[vertex] (y3) at (32.5,1) {};

\node[vertex] (neu) at (2.5,1) {};
\node[vertex] (n5) at (22.5,9) {};
\node[vertex] (n6) at (22.5,11) {};
\node[vertex] (n7) at (24,10) {};
\draw (n7)--(x2)--(n5)--(n6);

\draw (2)--(1)--(3)--(4)--(6)--(7)--(8)--(9)--(10)--(11)--(12)--(13); \draw (u1)--(5)--(4)--(v9)--(5);
\draw (3)--(6); \draw (8)--(10); \draw (10)--(12)--(y3)--(13)--(y2); \draw (11)--(x1);
\draw (v1)--(2)--(v3)--(1)--(v4)--(v5)--(v6)--(v7)--(v8)--(3)--(v4)--(v6)--(v8)--(v4);
\draw (1)--(v8); \draw (v2)--(2);
\draw (u2)--(u3)--(6)--(u4)--(u5)--(u6)--(u7)--(u8)--(7)--(u9)--(u2);
\draw (u3)--(u4)--(u8)--(u9)--(u3)--(u8)--(u6);
\draw (u3)--(7)--(u4)--(u9)--(6)--(u8); \draw (7)--(x2)--(x3);
\draw (9)--(x4)--(x5)--(x6)--(x7)--(x8)--(x9)--(y1)--(9);
\draw (y1)--(x4)--(x6)--(x8)--(y1);

\draw (neu)--(2);
\end{tikzpicture}
\caption{A $3$-block graph.}\label{fig:d-block}
\end{center}
\end{figure}

Let $d\ge 2$ and let $G$ be a $d$-block graph.
Recall that a block in $G$ is an \emph{endblock} if it contains at most one cut vertex. Vertices that are not cut vertices are called \emph{endvertices}. Thus, if $u$ is an endvertex then the block containing $u$ (which may or may not be an endblock) is unique.
We call a cut vertex $u$ a \emph{pseudo-endvertex} if $u$ belongs to at most $d-2$ endblocks and exactly one non-endblock.
%"at most one non-endblock": if u is in no non-endblock, so G has no d-claw et al!
Thus, for a pseudo-endvertex $u$, we say that $B$ is the unique block containing $u$, meaning that $B$ is the unique non-endblock that contains $u$. %Pseudo-endvertices will play an important role.

In computing an optimal $d$-claw deletion set for $G$, we first observe that there is a solution that does not contain any endvertex. It is by the fact that each block is $d$-claw free and if an endvertex $u$ is in a solution $S$ then it must be a leaf of some claw centered at a cut vertex $v$. Then $S-u+v$ is a solution that does not contain the endvertex $u$. The following lemmas are for pseudo-endvertices.

\begin{lemma}\label{lem:dblock-1}
%Let $u$ be an endvertex and let $B$ be the block containing $u$.
Let $u$ be a pseudo-endvertex.
Then any $d$-claw $C$ containing $u$, if any, is centered at a cut vertex $v\not=u$. Moreover,
\begin{itemize}
\item
if $B$ is the unique block containing $u$, then $C\cap B=\{u,v\}$;
\item
if $B'$ is an endblock containing $u$, then $C\cap B'=\{u\}$.
\end{itemize}
\end{lemma}
\begin{proof}
This is because every block is $d$-claw-free and the neighbors of any cut vertex in any block induce a clique.\qed
\end{proof}

An optimal $d$-claw deletion set for $G$ is also called \emph{solution}.

\begin{lemma}\label{lem:dblock-2}
There is a solution that contains no pseudo-endvertices.
\end{lemma}
%The proof of Fact~\ref{fact:dblock-2} is given in Appendix~\ref{app:fact}.
\begin{proof}
Let $S$ be a solution for $G$ and assume that $S$ contains a pseudo-endvertex~$u$.
Let $B$ be the unique block of $G$ containing $u$.
Since $S-u$ is not a $d$-claw deletion set, there is some $d$-claw $C$ of $G$ outside $S\setminus\{u\}$.
Then, of course,
\begin{equation}\label{eq:1}
C\cap S=\{u\}.
\end{equation}
By Lemma~\ref{lem:dblock-1}, the center $v$ of $C$ is a cut vertex of $G$ in $B$, and $C\cap B=\{u,v\}$. Thus,
for every $w\in N(v)\cap B$, $C-u+w$ is a $d$-claw, and by (\ref{eq:1}), $w\in S$. Hence
\begin{equation}\label{eq:2}
N(v)\cap B\subseteq S.
\end{equation}
We now claim that $S'=S-u+v$ is a $d$-claw deletion set (and thus $S'$ is a solution).
Indeed, let $C'$ be an arbitrary $d$-claw. If $u\not\in C'$ or $v\in C'$ then $C'\cap S'\not=\emptyset$.
So let us consider the case in which $u\in C'$ and $v\not\in C'$.
Then, by Lemma~\ref{lem:dblock-1}, the center $v'$ of $C'$ is a cut vertex of $G$ in $B$. Hence $v'$ and $v$ are adjacent, and by (\ref{eq:2}), $v'\in S'$.\qed
\end{proof}

We remark that Lemma~\ref{lem:dblock-2} is the best possible in the sense that a cut vertex $u$ belonging to exactly two non-endblocks may be contained in any solution; take the $d$-block graph that consists of two $d$-claws with exactly one common leaf $u$.

\begin{lemma}\label{lem:dblock-3}
Let $v$ be a cut vertex and let $B$ be a block containing $v$. If every vertex in $N(v)\cap B$ is a cut vertex, then $B=N[v]\cap B$. In particular, $B$ is a clique.
\end{lemma}
\begin{proof}
Suppose the contrary that $B\setminus N[v]\not=\emptyset$.
Then, as $B-v$ is connected, there is an edge connecting a vertex $w\in N(v)\cap B$ and a vertex $w'\in B\setminus N[v]$.
Now, as $w$ is a cut vertex, $N(w)\cap B$ is a clique, implying that $vw'$ is an edge. This is a contradiction, hence $B=N[v]\cap B$.\qed
\end{proof}

\begin{lemma}\label{lem:one-nonendblock}
If $G$ has at most one non-endblock, then a solution for $G$ can be computed in polynomial time.
\end{lemma}
%Due to space limitations, the proof is omitted.
%The proof of Lemma~\ref{lem:one-nonendblock} is given in Appendix~\ref{app:lemma}.
 \begin{proof}
 If all blocks of $G$ are endblocks, then $G$ has at most one cut vertex. In this case, $G$ contains a $d$-claw if and only if $G$ has at least $d$ blocks. If $G$ has at least $d$ blocks, then all $d$-claws in $G$ are centered at the unique cut vertex $v$ and $\{v\}$ is the solution.

 So, let $B$ be the block of $G$ that is not an endblock. Write
 \begin{itemize}
 \item $U'$ for the set of endvertices in $B$,	
 \item $U$ for the set of pseudo-endvertices in $B$,
 \item $X$ for the set of vertices in $B$ that belong to $\ge d$ endblocks,
 \item $Y$ for the set of vertices $v$ in $B$ that belong to exactly $d-1$ endblocks and $N(v)\cap U\not=\emptyset$,
 \item $Z=B\setminus (U'\cup U\cup X\cup Y)$.
 \end{itemize}
 Observe that $X\cup Y\cup Z$ is a $d$-claw deletion set for $G$, hence the size of an solution is at most $|X|+|Y|+|Z|$.

 By Lemma~\ref{lem:dblock-2}, there is a solution not containing any pseudo-endvertex. Such a solution $S$ must contain $X\cup Y$ because every vertex in $X\cup Y$ is the center of a $d$-claw in which all leaves are pseudo-endvertices.
 Thus, if $Z=\emptyset$, then $S=X\cup Y$.

 So, let us assume that $Z\not=\emptyset$. Note that, as $d\ge 2$, every vertex $v\in Z$ is a cut vertex, and by definition of $Z$, $N(v)\cap B$ contains no pseudo-endvertices. Hence, by Lemma~\ref{lem:dblock-3}, $B$ is a clique.
 Thus, by definition of $Z$ again, $U=\emptyset$, and therefore $Y=\emptyset$.
 Now, observe that at most one vertex in $Z$ may not be contained in $S$: $|Z\setminus S|\le 1$. Indeed, if $z_1$ and~$z_2$ were two vertices in $Z$ not belonging to~$S$, then $z_1, z_2$ and $d-1$ endvertices adjacent to~$z_1$ would together induce a $d$-claw outside $S$. Thus, $|S|\ge |X|+|Z|-1=|B|-1$.
 On the other hand, note that, for any $v\in B$, $B-v$ is a $d$-claw deletion set for~$G$.
 Thus, for any $v\in B$, $S=B-v$ is a solution.
 \qed
 \end{proof}

We are now ready to show that \mindclawVD{d} is polynomially solvable on $d$-block graphs.
Our proof is inspired by the polynomial result for \CVD\ on block graphs in \cite[Theorem 10]{CaoKOY17}. A \emph{block-cut vertex tree} $T$ of a (connected) graph~$G$ has a \emph{node} for each block of~$G$ and for each cut vertex of~$G$. There is an edge $uv$ in $T$ if and only if~$u$ corresponds to a block containing the cut vertex~$v$ of~$G$. It is well known that the block-cut vertex tree of a graph can be constructed in linear time.

\begin{theorem}\label{thm:dblock}
\mindclawVD{d} is polynomially solvable on $d$-block graphs.
\end{theorem}
%Due to space limitations, the proof is omitted.
%The proof of Theorem~\ref{thm:dblock} is given in Appendix~\ref{app:lemma}.
 \begin{proof}
 Let $T$ be the block-cut vertex tree of $G$. Nodes in $T$ corresponding to blocks in~$G$ are \emph{block nodes}. For a block node $u$ we use $B(u)$ to denote the corresponding block in~$G$. Nodes in $T$ corresponding to cut vertices in~$G$ are \emph{cut nodes} and are denoted by the same labels.

 Choose a node $r$ of $T$ and root $T$ at $r$. For a node $x\not= r$ of $T$, let $p(x)$ denote the parent of $x$ in $T$.
 Note that all leaves of $T$ are block nodes, the parent of a block node is a cut node, and the parent of a cut node is a block node. %For a cut node $v$, let ${\cal V}=\bigcup_{u} B(u)$ where $u$ runs over all children of $v$.

 Let $u$ be a leaf of $T$ on the lowest level and let $v=p(u)$ be the parent of $u$. Note that all children of $v$ correspond to endblocks in $G$, and if $r=p(v)$, then Lemma~\ref{lem:one-nonendblock} is applicable.
 So, assume $r\not=p(u)$ and write $u'=p(v)$, $v'=p(u')$.
%    % corresponds the unique non-endblock $B'=B(u')$ containing $v$.
%    %if $u'=r$, then Lemma 1 is applycable. Otherwise let $v'=p(u')$ and
 Note that by the choice of $u$, $B'=B(u')$ is the unique non-endblock containing vertices in $B'-v'$.

 If $v$ has at most $d-2$ children, then $v$ is a pseudo-endvertex in $G$. By Lemma~\ref{lem:dblock-2}, we remove $v$ and all children of $v$ from $T$.
%    %But v is still in the resulting graph G, as an endvertex!

 If $v$ has at least $d$ children, or $v$ has exactly $d-1$ children and some vertex in $N_G(v)\cap B'$ is a pseudo-endvertex, then put $v$ into the solution $S$ and remove $v$ and all children of $v$ from $T$. Correctness follows again from Lemma~\ref{lem:dblock-2}.
%    %Note that all children of $v$ correspond to endblocks in $G$,
%    %hence th correctness by Fact~\ref{fact:dblock-2}.

 It remains the case that $v$ have exactly $d-1$ children and no vertex in $N_G(v)\cap B'$ is a pseudo-endvertex.
 In particular, every vertex in $N_G(v)\cap B'$ is a cut vertex. Hence $B'$ is a clique by Lemma~\ref{lem:dblock-3}.
 Moreover, every vertex in $N_G(v)\cap (B'-v')$ belongs to at least $d-1$ endblocks.
 Now, note that all $d$-claws in $G$ containing $v$ contain a vertex in $B'-v$, and every solution for $G$ not containing pseudo-endvertices must contain at least $|B'|-1$ vertices in~$B'$.
 Thus, we put $B'-v$ into solution $S$ and remove the subtree rooted at~$v'$ from~$T$, and for each other child $u_i\not=u$ of~$v'$, we solve the problem on the subgraph induced by $B(u_i)$ %minus $v'$?
 and its children.
 Note that, by the choice of $u$, all these subgraphs satisfy the condition of Lemma~\ref{lem:one-nonendblock}.
 Finally, it is not hard to check that all of these checks can be done in linear time by working with the block-cut vertex tree $T$ of $G$.
 \qed
 \end{proof}

\section{Conclusion}\label{sec:conclusion}
%=========================================
This paper considers the $d$-claw vertex deletion problem, \dclawVD{d}, which generalizes the famous \VC\ (that is \dclawVD{1}) and the \CVD\ (that is \dclawVD{2}) problems and goes on with the recent study~\cite{Bonomo-Braberman20} on claw vertex deletion problem, \dclawVD{3}. 
It is shown that \CVD\ remains $\NP$-complete on bipartite graphs of maximum degree~3 and, for each $d\ge 3$, \dclawVD{d} remains $\NP$-complete on bipartite graphs of degree~$d$, and thus a complexity dichotomy with respect to degree constraint. 
It is also shown that \dclawVD{d} remains $\NP$-complete when restricted to split graphs of diameter~2 and to bipartite graphs of diameter~3 (with only two vertices of unbounded degree) and polynomially solvable on bipartite graphs of diameter~2, and thus another dichotomy with respect to diameter.   
We also define a new class of graphs called $d$-block graphs which generalize the class of block graphs and show that \dclawVD{d} is solvable in linear time on $d$-block graphs, extending the algorithm for \CVD\ on block graphs in~\cite{0001KOY18} to \dclawVD{d}, and improving the algorithm for (unweighted) \dclawVD{3} on block graphs in~\cite{Bonomo-Braberman20} to $3$-block graphs.

We note that \VC\ and \CVD\ have been considered by a large number of papers in the context of approximation, exact and parameterized algorithms. As a question for further research we may ask: which known results in case $d=1,2$ can be extended for all $d\ge 3$?
We believe that the approaches in~\cite{BoralCKP16,Tsur21} for \CVD\ could be extensible to obtain a similar parameterized algorithm for \dclawVD{d} for all $d\ge 3$.
Finally, recall that \VC\ is polynomially solvable on chordal graphs and \CVD\ is polynomially solvable on split graphs, and that \dclawVD{d} is $\NP$-complete on chordal graphs for $d\ge 3$.
Thus it would be interesting to clear the complexity of \CVD\ on chordal graphs~\cite{CaoKOY17,0001KOY18}.

%Other open question: \dclawVD{d} on planar graphs of maximum degree $d$?!

%===========================================

%\nocite{*}

%\bibliographystyle{amsplain} %{spbasic}
\bibliography{d-clawVD}

\end{document}